\pdfoutput=1
\documentclass{article}

\usepackage{microtype}
\usepackage{graphicx}
\usepackage{subfigure}
\usepackage{booktabs} 
\usepackage[utf8]{inputenc}

\usepackage{hyperref}
\usepackage{wrapfig}

 \usepackage[accepted]{icml2025}
\setcitestyle{numbers}

\usepackage{amsmath}
\usepackage{amssymb}
\usepackage{mathtools}
\usepackage{amsthm}
\usepackage{textcomp}

\usepackage[capitalize,noabbrev]{cleveref}

\theoremstyle{plain}
\newtheorem{theorem}{Theorem}[section]

\theoremstyle{definition}

\theoremstyle{remark}

\newcommand{\gnn}{~\textsc{SwaGen}}

\usepackage[textsize=tiny]{todonotes}

\icmltitlerunning{Enhancing Swarms’ Durability to Threats via GNN-based Generative Modeling}

\begin{document}
\newcommand{\fref}[1]{Figure~\ref{#1}}
\newcommand{\tref}[1]{Table~\ref{#1}}

\twocolumn[
\icmltitle{Enhancing Swarms' Durability to Threats via Graph Signal Processing and GNN-based Generative Modeling}



\icmlsetsymbol{equal}{*}

\begin{icmlauthorlist}
\icmlauthor{Jonathan Karin}{equal,cs}
\icmlauthor{Zoe Piran}{equal,cs,cur}
\icmlauthor{Mor Nitzan}{cs,phy,med}
\end{icmlauthorlist}

\icmlaffiliation{cs}{School of Computer Science and Engineering, The Hebrew University of Jerusalem}
\icmlaffiliation{cur}{Present address: Department of Computer Science, Stanford University, Stanford, CA, USA. Research and Early Development, Genentech, Inc., South San Francisco, CA, USA.}
\icmlaffiliation{phy}{Racah Institute of Physics, The Hebrew University of Jerusalem, Jerusalem, Israel}
\icmlaffiliation{med}{Faculty of Medicine, The Hebrew University of Jerusalem, Jerusalem, Israel}


\icmlcorrespondingauthor{Mor Nitzan}{mor.nitzan@mail.huji.ac.il}

\icmlkeywords{Machine Learning, ICML}

\vskip 0.3in
]



\printAffiliationsAndNotice{\icmlEqualContribution} 

\begin{abstract}
Swarms, such as schools of fish or drone formations, are prevalent in both natural and engineered systems. While previous works have focused on the social interactions within swarms, the role of external perturbations--such as environmental changes, predators, or communication breakdowns--in affecting swarm stability is not fully understood. Our study addresses this gap by modeling swarms as graphs and applying graph signal processing techniques to analyze perturbations as signals on these graphs. By examining predation, we uncover a ``detectability-durability trade-off", demonstrating a tension between a swarm's ability to evade detection and its resilience to predation, once detected. 
We provide theoretical and empirical evidence for this trade-off, explicitly tying it to properties of the swarm's spatial configuration.  
Toward task-specific optimized swarms, we introduce \gnn, a graph neural network-based generative model. We apply \gnn~to resilient swarm generation by defining a task-specific loss function, optimizing the contradicting trade-off terms simultaneously.With this, \gnn~reveals novel spatial configurations, optimizing the trade-off at both ends. Applying the model can guide the design of robust artificial swarms and deepen our understanding of natural swarm dynamics. 
\end{abstract}
\section{Introduction}\label{sec:intro}

\begin{figure*}[h!] \label{fig:abstract}
	\centering
	\includegraphics[width=0.85\linewidth]{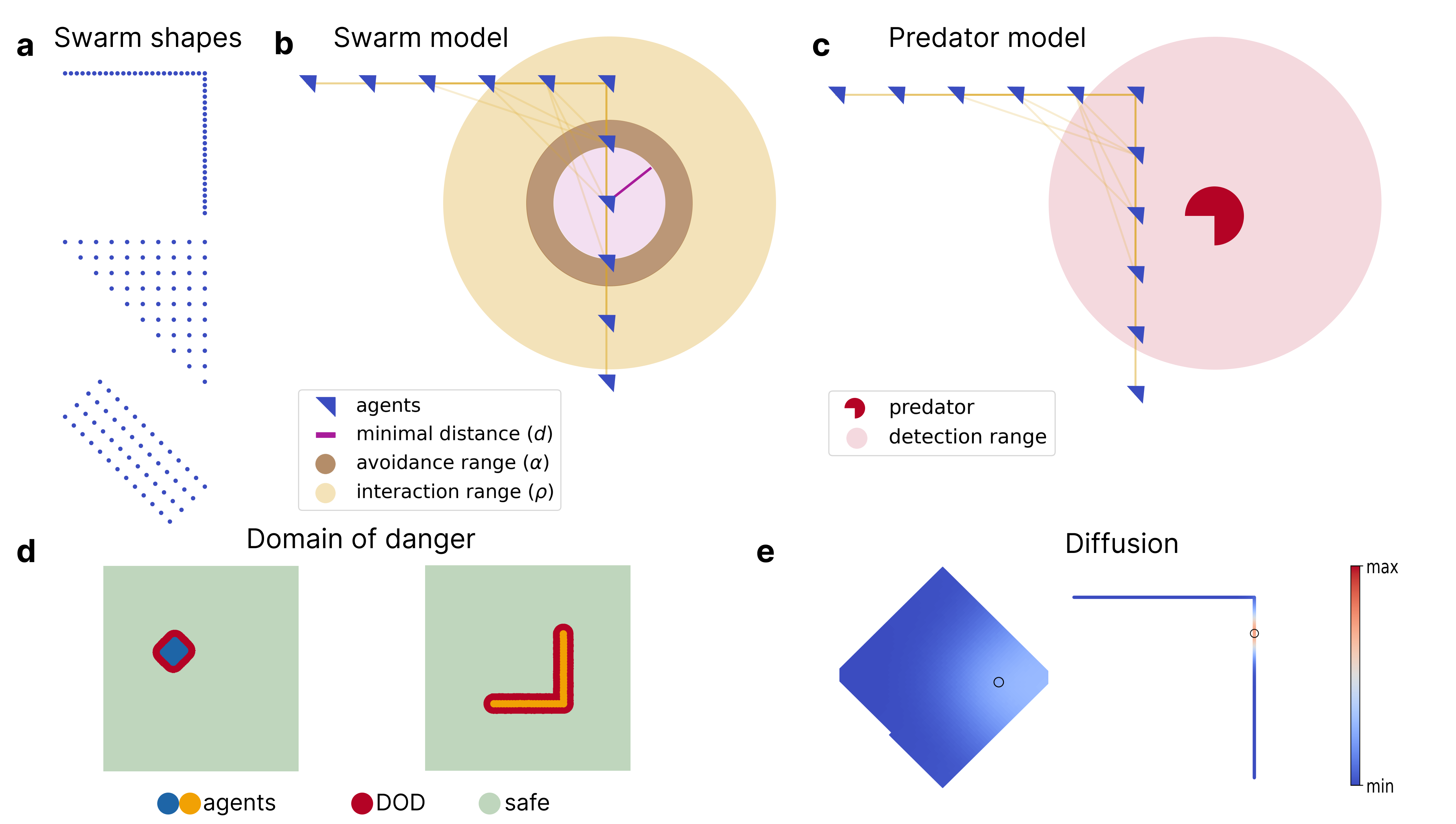}
	\caption{\textbf{An abstract visualization of the GSP approach used to model swarm behavior.} \textbf{a,} Initial swarming configuration used in our framework: v-formation (top), arrow (middle), rectangle (bottom). \textbf{b,} The swarming model (see~\ref{sec:swarm-model}), along with the introduced graph constructions. Agents are in blue triangles, and orange lines represent the induced swarm graph (edges added in accordance with the interaction range). Characteristic parameters shown for a representative agent; the avoidance range ($\alpha$; brown), the interaction range ($\rho$; yellow), and the minimum distance ($d$; pink). \textbf{c,} The predation model (see~\ref{sec:predator-model}). The predator, a red shape, and the detection range in light red. \textbf{d-e,} Properties of the rectangle (left) and v-formation (right) considering $N=2000$ agents; (d) the diffusion ability visualized using a signal propagated through heat filter (with $\tau=50$), and (e) the domain of danger (red).}
 \vspace{-3mm}
\end{figure*}
Collections of agents that move synchronously in close spatial proximity are omnipresent in biology and engineering. These collections, termed ``swarms'', are composed of agents that may be living organisms, such as birds or school of fish, or artificial entities, such as drones. Swarm behavior is affected by the collective decision-making process. These movement-related decisions can result from a ``democratic'' process, where the majority opinion prevails~\cite{strandburg2015shared}, or could be guided by a few informed individuals, with a specific moving direction tendency~\cite{couzin2005effective}.

Swarm-dynamics models often assume that the agent's decisions rely on its local neighborhood. Hence, in these models, location updates are often guided by averaging the dynamics within the local neighborhood~\cite{sridhar2021geometry,couzin2005effective, strandburg2015shared}.
For example, in the swarming model introduced by~\cite{couzin2005effective}, each agent aligns its movement with its neighbors while maintaining a minimum distance away from its closest neighbor (\fref{fig:abstract}), intended to avoid collisions. These models were shown to capture qualitative and quantitative measured features of collective swarm motion~\cite{couzin2005effective, sridhar2021geometry}, which indicates that local interactions are sufficient to describe natural collective behavior, allowing groups to navigate and forage effectively~\cite{couzin2005effective, couzin2003self}.

Previous works assessed potential advantages of different swarm spatial configurations, such as v-formations of birds during long-distance travel, which promote conservation of energy by utilizing the up-wash from neighboring birds (\fref{fig:abstract}a,~\citealp{cattivelli2011modeling}). 
In addition, the implication of the swarm's spatial configuration on its robustness to external perturbations has been explored;~\cite{mateo2017effect} evaluate the relationship between the number of neighbors interacting with an agent and the time between two consecutive predator catches, and~\cite{mateo2019optimal} investigate the influence of the network topology on the collective response of the swarm. 
However, an inclusive framework that links between a swarms' spatial configuration and observed collective behavior—allowing for task-specific modeling, analysis, and optimization—is lacking. 

We can model the spatial configuration of a swarm using a graph. In this graph, agents are represented by nodes, and edges indicate which agents are close enough to interact. Next, with graph signal processing (GSP), an emerging field that extends traditional signal processing techniques to data defined on irregular structures, represented as graphs~\cite{ortega2018graph}, we can incorporate external and internal signals encountered by the swarm, and its effective response to them. Specifically, GSP enables the analysis and manipulation of signals with complex dependencies, which make it suitable for understanding self-organizing systems~\cite{dong2020graph}. In such systems, local interactions between components lead to the emergence of global behavior, often observable in biological, social, and technological networks. In the context of swarm modeling, GSP approaches were introduced to find behaviorally anomalous agents~\cite{schultz2021detecting} and to analyze the spectral properties of basic swarm topologies~\cite{schultz2021analyzing}. 

Graph neural networks (GNNs) extend GSP by introducing learnable filters. They operate on graph-structured data, using message-passing mechanisms to aggregate information from local neighborhoods~\cite{zhou2020graph}. Applications of GNNs in the field of swarming dynamics include learning decentralized controllers to optimize collective behavior~\cite{tolstaya2020learning} and designing communication strategies for improved coordination~\cite{chen2023spatial}.

In this study, we propose a GSP-based method to link the collective behavior of swarms to their spatial configuration. Our framework aims to quantify the effect of external cues as a function of the geometric structure of the swarm and generate task-specific swarm configuration. We demonstrate the application of this framework to study the response of swarms to external perturbations in the form of predation--exposing an inherent ``detectability-durability trade-off". The trade-off is studied both theoretically and empirically, presenting the relationship between a swarm's resilience and spatial configuration parameters. Using the modeling framework, we suggest \gnn, a task-specific GNN-based generative model for optimized swarm configurations. Applying \gnn~to the task of durable swarm design we identify novel optimized configurations, improving swarm's resilience.

\section{Preliminaries} \label{sec:preliminaries}

\subsection{Swarming model}\label{sec:swarm-model}
Couzin et al.~\cite{couzin2005effective} presented a swarming model driven by social interactions, where the agents' movement direction is modified based on its local neighborhood. In this setting, a groups consist of $N$ individuals, 
where each individual $i$, is associated with a position vector $\mathbf{c}_i(t)$, a direction vector $\mathbf{d}_i(t)$, and a constant speed value $s_i$. The initial locations are typically initialized to follow a well-defined structure, considering common classes observed in nature, e.g., v-formation, arrow, or rectangle (\fref{fig:abstract}a, ~\citealp{portugal2020bird,vicsek2012collective}). The location of each individual $i$ is then updated in time, 
$\mathbf{c}_i(t)= \mathbf{d}_i(t)* s_i$.
To avoid collisions, each agent $i$ attempts to maintain a minimum \textit{avoidance range}, $\alpha$ (\fref{fig:abstract}b) from others $j$ by turning away from neighbors within that range.
Additionally, each agent is attracted to, and aligned, relative to neighbors located outside the minimum distance but within the \textit{interaction range}, $\rho$ (\fref{fig:abstract}b). Together the updated direction of agent $i$ is given by:
\begin{align*}
    \mathbf{d}_i(t + \Delta t) = \sum_{j \in \{j : \alpha \leq |\mathbf{c}_j(t) - \mathbf{c}_i(t)| < \rho\}} &\frac{\mathbf{c}_j(t) - \mathbf{c}_i(t)}{|\mathbf{c}_j(t) - \mathbf{c}_i(t)|}  \\
   +  &\frac{\mathbf{d}_j(t) - \mathbf{d}_i(t)}{|\mathbf{d}_j(t) - \mathbf{d}_i(t)|} \ .
\end{align*}

As suggested in~\cite{couzin2005effective}, to enable control over the movement of the swarm or incorporate external information, a fraction $p$ of agents are provided with a preferred direction  $\hat{\mathbf{g}}$ (a unit vector in the direction to e.g. a resource or migration route). As opposed to this informed subset of agents, naive individuals have no directional preference. Additionally, none of the agents are aware of which group members are informed.
Informed individuals incorporate  $\hat{\mathbf{g}}$ into their preferred direction choice, modifying $\hat{\mathbf{d}}_i(t + \Delta t)$ (the unit vector in the desired direction) to $\mathbf{d}_i'(t + \Delta t)$:
\begin{align*}
\mathbf{d}_i'(t + \Delta t) = \hat{\mathbf{d}}_i(t + \Delta t) + \omega \cdot \hat{\mathbf{g}_i} \ ,
\end{align*}
where $\omega$ dictates how much weight is given to the external information provided by $\hat{\mathbf{g}_i}$. If $\omega = 0$, 
the preferred direction $\hat{\mathbf{g}_i}$ of informed individual $i$ has no effect on its direction vector (turning these effectively into naive individuals).
As $\omega$ increases, individuals balance and then overweigh ($\omega > 1$) their preferred direction relative to social interactions.

\subsection{Predator model}\label{sec:predator-model} 
Predation is an example of an external perturbation to swarms' collective behavior, and can be modeled using the theory of marginal predation~\cite{hamilton1971geometry,morrell2011spatial}. The assumption here is that, at each time step, predators approach and attack the closest potential prey located within a fixed and finite distance $\delta$ (\textit{detection range}, \fref{fig:abstract}c,  ~\citealp{james2004geometry}). 
To simplify, we will use the same value for both detection and interaction range ($\delta=\rho$,  see~\ref{sec:swarm-model}).
The predator moves in a constant, random, direction, until an agent (prey) enters its detection range--marking a detection event. Once a detection occurs, the predator changes its direction to advance towards the agent. An encounter between an agent and a predator is considered a predation event.
The direction vector of the predator is updated using:
\begin{align*}
    \mathbf{d}_{\text{p}}(t + \Delta t) &= 
    \begin{cases} 
        \frac{\mathbf{c}_i(t) - \mathbf{c}_{\text{p}}(t)}{|\mathbf{c}_i(t) - \mathbf{c}_{\text{p}}(t)|} &  |\mathbf{c}_i(t) - \mathbf{c}_{\text{p}}(t)| < \rho \ , \\ 
        \text{random unit vector} & \text{otherwise,}
    \end{cases}
\end{align*}
where $\mathbf{c}_{\text{p}}(t)$ is the position vector of the predator at time $t$ and $\mathbf{c}_i(t)$ is the position vector of agent $i$ at time $t$.
Here $i$ is the index of the closest prey agent within the detection range.

As suggested by~\cite{olson2016evolution}, to allow for predation, the speed of a predator must be larger than the speed of a non-predator agent.
Additionally, each agent $i$ aims to avoid a predator when it enters the agent's interaction range $\rho$, by moving in the opposite direction of the predator~\cite{chen2014minimal}:
\begin{align*}
    \mathbf{d}_i(t + \Delta t) &=  - \frac{\mathbf{c}_{\text{p}}(t) - \mathbf{c}_i(t)}{|\mathbf{c}_{\text{p}}(t) - \mathbf{c}_i(t)|} \ ,
\end{align*}

\subsection{Graph signal processing}\label{sec:gsp}
Graph Signal Processing (GSP) is a framework that extends classical signal processing concepts to signals defined on irregular structures represented by graphs~\cite{ortega2018graph}. A graph $ G = (V, E) $ is defined by a set of vertices $( V )$ and a set of edges $( E )$ connecting pairs of vertices. The adjacency matrix $( \mathbf{A})$ of a graph is an $ N \times N $ matrix where $ N $ is the number of vertices. The element $ A_{ij} $ is non-zero if there is an edge between vertex $ i $ and vertex $ j $.

The graph Laplacian $ \mathbf{L} $ is a key operator in GSP~\cite{ortega2018graph}, defined as $ \mathbf{L} = \mathbf{D} - \mathbf{A} $, where $ \mathbf{D} $, termed the degree matrix, is a diagonal matrix with $D_{ii}$ the degree of vertex $i$. 
The Laplacian can be used to capture various aspects about the structure and features of its associated graph. One such feature is the total variation $\mathrm{TV (\mathbf{v})};$ of a graph signal 
$\mathbf{v} \in \mathbb{R}^N$, measuring the smoothness of $\mathbf{v}$ over the graph (with lower values indicating a smoother signal):
\begin{align*}
\mathrm{TV}(\mathbf{v}) = \mathbf{v}^\top \mathbf{L} \mathbf{v} = \sum_{i,j} A_{ij} {(v_i - v_j)}^{2}
\ . 
\end{align*}
The Graph Fourier Transform (GFT) of signal $\mathbf{v}$ is described as follows: 
\begin{align*}
    \hat{\mathbf{v}} = \mathbf{U}^\top \mathbf{v} \ ,
\end{align*}
where $\mathbf{U}$ is the matrix whose rows are the eigenvectors of the graph Laplacian $\mathbf{L}$.
In GSP, graph filters process signals by amplifying specific frequency components~\cite{shuman2016vertex}. Low-pass filters, such as the heat filter, smooth signals by preserving low-frequency components and diffusing the signal across the graph~\cite{shuman2016vertex}. Such filters reduce the total variation of the signal and are essential for tasks like denoising and smoothing~\cite{ortega2018graph}.


\subsubsection{Diffusion on graphs} \label{sec:diff_graph}

Diffusion processes on graphs are used to model the spread of information or substances through the vertices of a graph~\cite{hammond2013graph}. 
The diffusion equation on a graph, describing the evolution of the signal over time, can be written as:
\begin{align*}
     \frac{d\mathbf{h}\left(\tau \right)}{d\tau} = -\mathbf{L}\mathbf{h}\left(\tau \right) \ ,
 \end{align*}
 where $\mathbf{h}\left(\tau \right)$ is a signal on the graph at time $\tau$. 
The solution is found by exponentiating the Laplacian eigenspectrum, and by that define the heat kernel $\mathbf{K}\left(\tau \right) \in \mathbb{R}^{n \times n}$:

 \begin{align*}
     \mathbf{K}\left(\tau \right) = \sum_{k}e^{-\tau \lambda_{k}} \mathbf{u}_{k}\mathbf{u}^{\top}_{k}\ = \mathbf{U} e^{-\tau \Lambda} \mathbf{U}^\top
 \end{align*}
The solution of the diffusion equation is then given by:
 \begin{align*}
     \mathbf{h}\left(\tau \right)  = \mathbf{K}\left(\tau \right)\mathbf{h}(0) = \mathbf{U} e^{-\tau \Lambda} \mathbf{U}^\top\mathbf{h}(0)
 \end{align*}
 Where $\Lambda$ is a diagonal matrix, with $\Lambda_{kk} = \lambda_k$, $\lambda_k$ is the $k$'th eigenvalue of $\mathbf{L}$, and $\mathbf{h}(0)$ is the initial condition.
 
An entry in the kernel denotes the signal propagation between two nodes in the graph, namely: 
 \begin{align*}
     \mathbf{K}_{i,j}(\tau) = \sum_{k}e^{-\tau \lambda_{k}} \mathbf{u}_{k}(i)\mathbf{u}_{k}(j) 
 \end{align*}

When the graph is embedded in a manifold that is locally euclidean, the signal propagation can be approximated by the Gaussian~\cite{bai2004heat}: 
\begin{align*} 
    \mathbf{K}_{i,j}\left(\tau \right) =\left(4\pi \tau\right)^{-n/2} e^{- \frac{d\left(i,j \right)^{2}}{4 \tau^2}} \ ,
\end{align*}

where $d\left(i,j \right)$ is the distance between nodes $i$ and $j$ on the manifold and $n$ is the number of nodes in the graph.
This Gaussian approximation directly links diffusion time $\tau$ to distance $d\left(i,j \right)$. For the signal to be non-negligible (i.e. pass between nodes) at a large distance, the time $\tau$ must also be correspondingly large, specifically scaling such that $\tau \sim d\left(i,j \right)$.

Lastly, given a graph impulse signal $\delta_i$ ('1' at vertex $i$ and '0' elsewhere), the impulse response for the heat kernel at vertex $j$ can be approximated by:
\begin{align*}
    (\mathbf{K} * \delta_i)_j 
    &= \sum_{k} \mathbf{K}_{j,k}(\tau) \, \delta_i(k) \\
    &= \mathbf{K}_{j,i}(\tau) \\
    &= \left(4\pi \tau\right)^{-n/2} \exp\left(- \frac{d(i,j)^2}{4\tau^2} \right)
\end{align*}

\subsection{Graph Neural Networks}
Graph Neural Networks (GNNs) are a class of neural networks that extend graph signal processing by introducing learnable filters that learn local structures~\cite{zhou2020graph}. GNNs iteratively aggregate and transform features from neighboring nodes. A key property of GNNs is permutation invariance: the learned embeddings remain unchanged under the reordering of nodes~\cite{wu2020comprehensive}. Therefore, GNNs are a suitable candidate model for learning local motifs for swarms.
For the GNN layer we used GraphSAGE~\cite{hamilton2017inductive} as it concatenates self-node features with aggregated neighborhood features, preserving node-specific information, formally: 
\[
\mathbf{h}_v^{(0)} = \mathbf{x}_v, \quad \forall v \in V
\]
\[
\mathbf{h}_{\mathcal{N}(v)}^{(k)} = \text{AGGREGATE}_k\big(\{\mathbf{h}_u^{(k-1)} : u \in \mathcal{N}(v)\}\big), 
\]
\[
\mathbf{h}_v^{(k)} = \sigma\Big(\mathbf{W}_k \cdot \text{CONCAT}\big(\mathbf{h}_v^{(k-1)}, \mathbf{h}_{\mathcal{N}(v)}^{(k)}\big)\Big)
\]
\[
\mathbf{h}_v^{(k)} = \frac{\mathbf{h}_v^{(k)}}{\|\mathbf{h}_v^{(k)}\|_2}, \quad \forall v \in V
\]
\[
\mathbf{z}_v = \mathbf{h}_v^{(K)}, \quad \forall v \in V
\]
$\mathbf{h}_{v}^{(k)}$: embedding of node $v$ at layer  $k$.  $\mathbf{x}_v$: input features of node $v$
$\mathcal{N}(v)$: set of neighbors of node $v$, $\mathbf{h}_{\mathcal{N}(v)}^{(k)}$: aggregated embedding from $\mathcal{N}(v)$ at layer $k$.
$\sigma$: Non-linear activation.
$\mathbf{W}_k$: learnable weight matrix at layer $k$. $\mathbf{z}_v$: output embedding of node $v$.

\section{Linking spatial configuration and task response}\label{sec:analytical_model}

\begin{figure}[htb!]
 \centering
\includegraphics[width=0.9\linewidth]{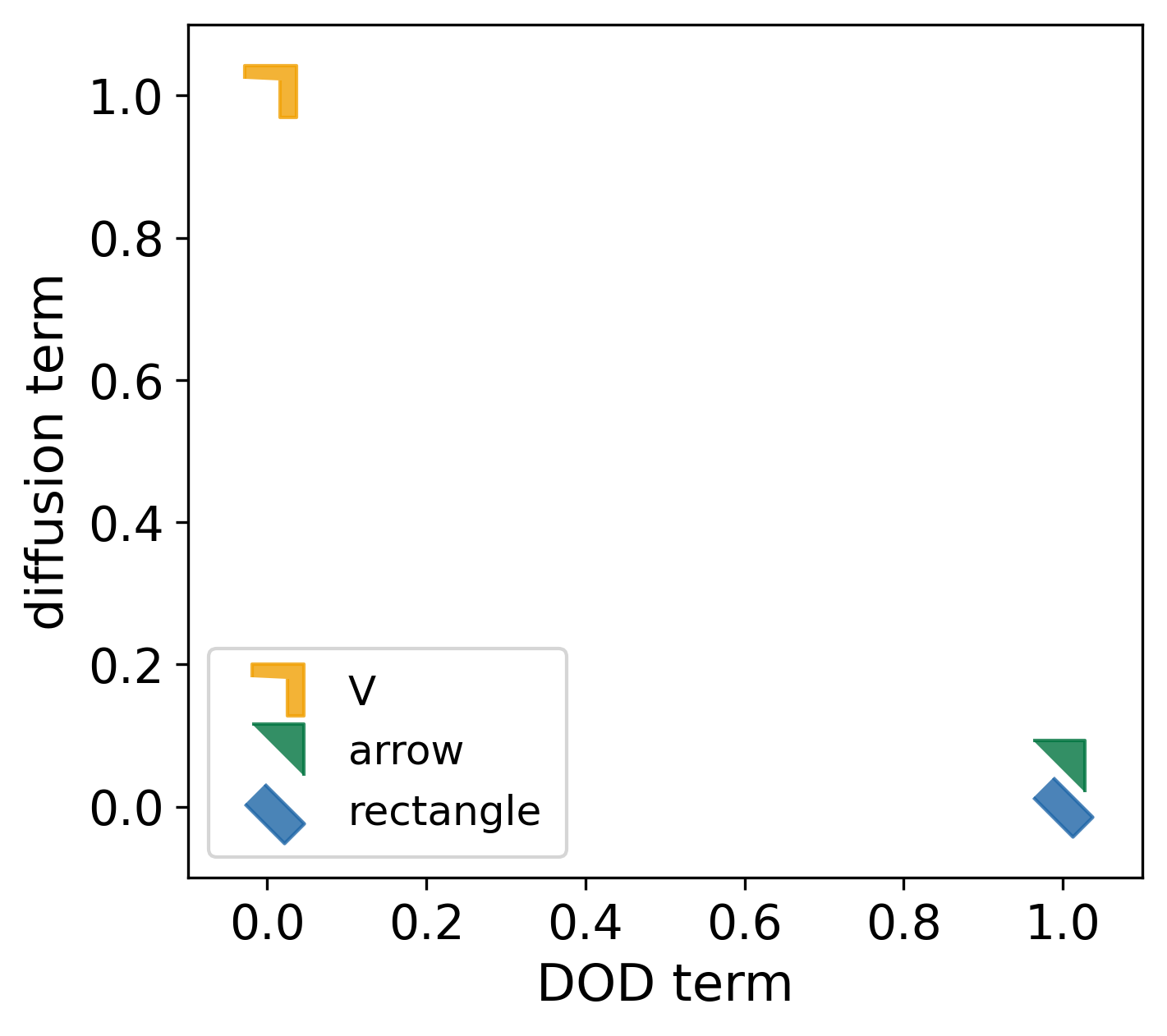}
  \caption{\textbf{The DOD-diffusion plane.} The tension between the domain of danger (DOD) and diffusion ability in common spatial configurations of swarms. Evaluation of the DOD (x-axis) and diffusion term (y-axis) normalized to the range of $[0,1]$. Evaluation for each spatial configuration is performed using $N=1000$ agents and $d=5$ (the distance between neighboring agents, see~\ref{app:simulation-details}).
To evaluate the diffusion we computed the mean values of the $CV$ for every possible source agent each configuration and considered heat kernels with $\tau=[50,100,150]$.}
  \label{fig:shapes}
  \vspace{-2mm}
\end{figure}

We first represent the swarm as a graph which captures the swarm's spatial configuration and provides a convenient modeling framework for internal and external signals.
Formally, the swarm is modeled as a graph \( G = (V, E) \), where each agent is represented by a node \( v_i \in V \), and an edge \( (v_i, v_j) \in E \) exists if the Euclidean distance between agents \( i \) and \( j \), denoted by \( \| \mathbf{x}_i - \mathbf{x}_j \| \), is less than the interaction range \( \rho \)(\fref{fig:abstract}b):
\[
|| \mathbf{x}_i - \mathbf{x}_j \| < \rho \implies (v_i, v_j) \in E \ .
\]

Using this framework we analyze the impact of the geometric properties of a swarm graph on its durability to a predator's attack (assuming the predator moves sequentially from one prey to the next,~\citealp{morrell2011spatial}). Note that this modeling can be easily extended to additional perturbations, e.g. air currents along the route of bird flocks flight~\cite{van2020environmental, calovi2015collective}, or adversarial attacks of artificial swarms, disrupting information passing between agents~\cite{primiero2018swarm}.
To explore this configuration-durability dependence we focus on common classes of spatial configurations of swarms: v-formation, arrow, and rectangle (\citealp{portugal2020bird,vicsek2012collective}, \fref{fig:abstract}d,e, \fref{fig:shapes visualization}a,b,c, see~\ref{app:initial-shape}). 

The predation process can be represented as two consecutive stages, which are modeled as signals in the graph (see~\ref{sec:predator-model}). First, the predator is at a \textit{detection phase}, posing a risk to the entire swarm. The detection phase terminates in a detection event, which occurs once the predator is within distance $\rho$ (the detection range) from an agent. Hence, the union of disks of radius $\rho$, centered at agents in the swarm, is considered the swarm's \textit{domain of danger} (DOD;~\citealp{james2004geometry}, \fref{fig:abstract}d). We formally prove that the magnitude of the DOD (provided as a fraction of the total area of the grid) dictates the probability of detection of the swarm (see Appendix~\ref{sec:detect}). 

Following detection the predator starts its \textit{predation phase}, advancing from one agent to another within the swarm. This phase is modeled as a diffusible signal on the graph, that is, a signal that spatially propagates from a source node, the first detected agent. Thus, we can measure a predator's effectiveness in hunting the swarm by evaluating the diffusion capability in the swarm's graph representation. 
Diffusion in a graph can be modeled by defining a source signal, mimicking a local perturbation at location $i$, and considering the signal's propagation along the graph using a heat filter~\cite{defferrard2017pygsp}. The coefficient of variation ($CV = \sigma/\mu$, with $\sigma$ the standard deviation and $\mu$ is the mean) of the signal after passing through the heat filter quantifies the ability of the signal to diffuse along the graph, with lower values indicating more efficient diffusion of the signal throughout the graph (\fref{fig:abstract}e). We provide the theoretical foundation and empirical evidence for the correlation between the duration of the predation phase and the ability of a signal to diffuse in the graph, by tying it to shortest path lengths along the graph using the Gaussian heat kernel approximation (see~\ref{sec:diff_graph}, Appendix~\ref{sec:durability}).

The above implies that a swarm whose associated graph structure has a smaller DOD is harder to detect, whereas graphs which are harder to diffuse in, will be more resilient to a predator's attack once detected, making them less likely to face extinction.

These properties present an inherent contradiction: a smaller DOD implies a more compact spatial configuration that enhances diffusion, while a larger DOD suggests a less compact configuration that impedes it (~\fref{fig:shapes},~\tref{table:mean_cv}, Appendix~\ref{sec:tradeoff}). 
This tension defines the basis for the ``detectability-durability trade-off", presented in the following section.

\subsection{The detectability-durability trade-off}\label{sec:res-predation}

We turn to address the predation durability task using the derived framework--we simulate swarms, with two predators whose initial positions are selected randomly across the simulated grid, and evaluate the population size at the end of the simulation (\fref{fig:predator}a, see~\ref{sec:predator-model}, \ref{app:simulation-details}).

The mean percentage of surviving agents is relatively stable between spatial configurations (v-formation: $66.5\pm34.8\%$, arrow: $64.4\pm47.8\%$, rectangle: $67.6\pm46.8\%$). However, the distribution reveals a dependence on the shape--exposing characteristics associated with each spatial configuration.

The detectability-durability plane reveals the expected trade-off between the percentage of simulations in which the swarm avoided detection and the percentage in which it avoided extinction (\fref{fig:predator}b).
This analysis confirms that detection is associated with the swarm's DOD, while durability relates to its diffusion properties. A comparison between the two trade-off planes
(\fref{fig:shapes}, \fref{fig:predator}b) demonstrates that that the model and simulation results recover a consistent relative ordering of the spatial configurations.
For example, we observe that the least compact shape (maximal DOD term), v-formation, is easiest to detect (\%58.0 of the simulations compared to \%36.0 and \%32.4 for the arrow and rectangle, respectively); however, its spread across the grid, making it harder to diffuse in it, minimizes its extinction (\%5.6 of the simulations compared to \%34.0 and \%32.0 for the arrow and rectangle, respectively). 
This result is stable across different simulation settings including scenarios with stochastic predation (predation occurs at probability $p$), stochastic interaction range, and stochastic speed (both sampled in each iteration from a Gaussian distribution centered at parameter values used for the deterministic model, ~\fref{fig:app_more_models}).

\begin{figure}[htb!]
 \centering
\includegraphics[width=0.85\linewidth]{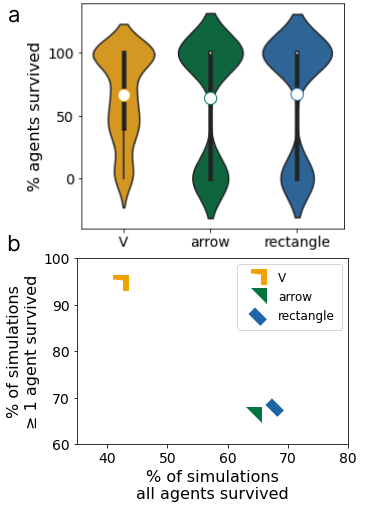}
  \caption{\textbf{Swarm's susceptibility to predation as a function of spatial configurations.} Results based on $250$ simulations for each configuration (v-formation, arrow, rectangle). \textbf{a,} Violin plots of the number of living agents at the end of the simulation for different initial shape structures: v-formation, an arrow, and a rectangle. 
  \textbf{b,} The detectability-durability plane, the percentage of simulations in which the swarm avoided extinction (y-axis) vs. the percentage in which it avoided detection (x-axis).
  }
  \label{fig:predator}
\vspace{-2mm}
\end{figure}

\section{Task-specific optimal configurations} \label{sec:generative}

The modeling framework sets the ground for generating task-specific optimal swarm configurations. A GNN-based architecture suits our setting as it uses message passing between neighboring nodes (agents), similarly to the manner in which interactions are defined in the swarming model (see~\ref{sec:swarm-model}). In addition, the permutation invariance property of GNNs ensures that representations are independent of the ordering of agents within the swarm--a desirable property to induce local structural motifs. Lastly, GNNs allow varying the swarm size, generalizing across the number of agents.

Hence, we propose \gnn, a GNN-based generative model that optimizes the initial structure of the swarm.
The architecture of \gnn is task-agnostic, capturing the spatial configuration of the swarm, however, the loss of \gnn~needs to be modified with respect to the task (\fref{fig:gen_architecture}). The loss contains two basic components, acting as regularizes and ensuring validity of the spatial configurations (e.g. ensure that swarm's agents are located within the predefined grid limits). To achieve task-specific optimization, these core components must then be supplemented with terms dedicated to the specific task.  

In the case of predation attack, we focus on the task of maximizing the number of surviving agents.
Relying on the established connection between a swarm's spatial configuration and resilience, the detectability-durability trade-off, we define the spatial configuration task over the DOD-diffusion trade-off.  Using such a generative approach uncovers a diverse set of solutions that represent a variety of resilient swarm configurations, that may not be intuitive or easily designed manually.

\begin{figure*}[htb!]
  \begin{center}
    \includegraphics[width=0.8\linewidth]{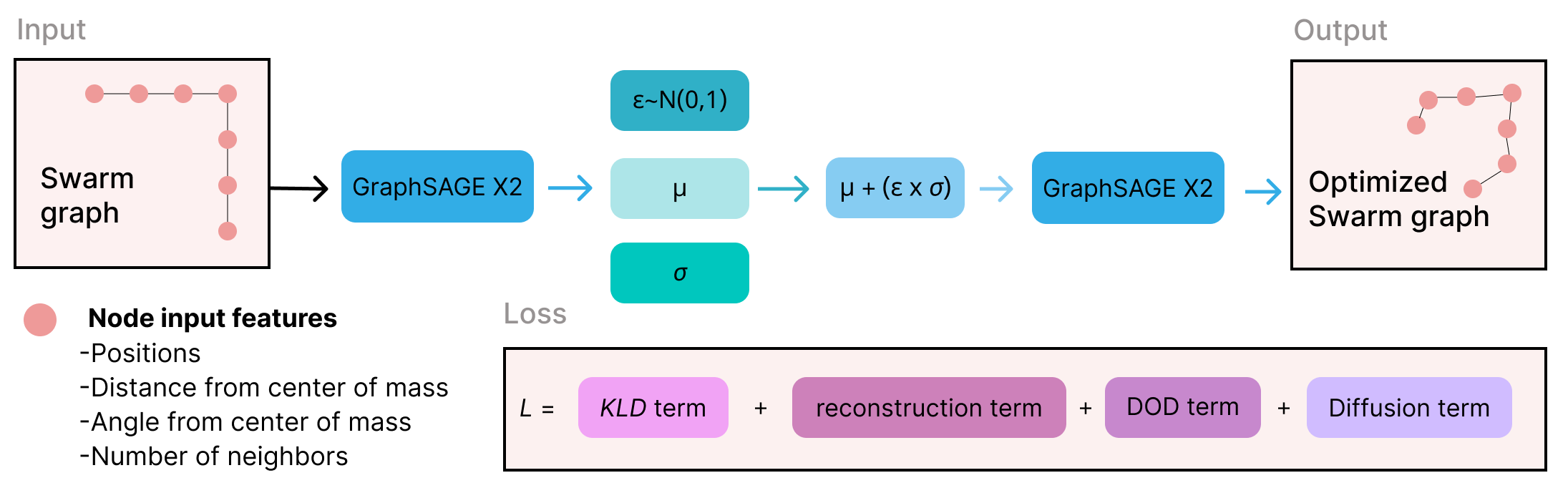}
  \end{center}
  \caption{\textbf{Overview of the \gnn~model architecture.}  The \gnn~model input is a swarm represented as a graph with $N$ neighbors and the following per-node features: position, distance and angle from the swarm's center of mass. The graph is processed through two GraphSAGE~\cite{hamilton2017inductive} layers, followed by a sampling layer, and two GraphSAGE layers. The loss function consists of four terms:  KL divergence, reconstruction, DOD, and diffusion. The output is an optimized swarm graph.}
  \label{fig:gen_architecture}
\end{figure*}

\subsection{Generative model description}\label{sec:gen-model-des}
As visualized in~\fref{fig:gen_architecture}, \gnn~ consists of four graph convolution layers (GraphSAGE,~\cite{hamilton2017inductive}), interleaved with an intermediate sampling layer, used to model the latent space. All layers, apart from the first, use ReLU activation. For the first layer, a Sine function is used to improve the learnability of high-frequency features~\cite{sitzmann2020implicit}. The model's input is the swarm positions ($C \in \mathbb{R}^{N \times d}$, where $N$ is the number of agents and $d$ is the space dimension), the number of neighbors of each agent, and the distance and angle of each agent relative to the swarm’s center of mass (see~\ref{app:model-input}). As a graph input for the GNN we use the swarm graph $G = (V, E)$ (see~\ref{app:graph-construction}). The output is an optimized arrangement of the swarm ($\Tilde{C} \in \mathbb{R}^{N \times d}$).
The loss function consists of four terms whose motivation is to approximate in a differentiable way the principles uncovered in the prior analysis (see~\ref{sec:analytical_model}).
\begin{enumerate}
    \item \textbf{Kullback-Leibler Divergence (KLD) term ($\mathcal{L}_{\text{KLD}}$)}, ensures that the learned latent space approximates a predefined Gaussian distribution~\cite{pinheiro2021variational}:
    \[
    \mathcal{L}_{\text{KLD}} = -\frac{1}{2} \sum \left( 1 + \log(\sigma^2) - \mu^2 - \sigma^2) \right),
    \]
    where $\mu$ is the mean and $\sigma^2$ is the variance of the latent distribution.
    \item \textbf{Reconstruction term ($\mathcal{L}_{\text{rec}}$)}, the mean squared error between the initial agents' position and the optimized agents' position:
\[
\mathcal{L}_{\text{rec}} = \text{MSE}(C, \tilde{C})
\]
\item \textbf{Domain of Danger (DOD) term ($\mathcal{L}_{\text{DOD}}$)}, the area surrounding the swarm, defined by the union of disks surrounding all agents (\fref{fig:abstract}d, see~\ref{sec:analytical_model}). We aim to minimize the DOD and thereby reduce the predator's effective detection range. For smoothness and differentiability, we use the following approximation:
\[
\mathcal{L}_{\text{DOD}} = \\
  \sum_{x, y} \sigma( \sum_{i=1}^{N} \sigma(-k (\sqrt{(x - \tilde{C}_{i,x})^2 + (y - \tilde{C}_{i,y})^2} -r)))\\
\]
where $N$ is the number of agents, sum is taken from all discrete points on the grid $(x, y)$. $C_{i,x}$ and $C_{i,y}$ denote the $x$- and $y$-coordinates of agent $i$, and $r$ is the interaction range. $\sigma$ is the sigmoid function $\sigma(x) = {1}/{1 + e^{-x}}$, and the constant $k$ controls its slope. 

\item \textbf{Diffusion term ($\mathcal{L}_{\text{diff}}$)}, aims to limit the predator's ability to spread within the swarm and predate all agents. For that, we aim to minimize the swarm's graph diffusion ability (\fref{fig:abstract}e).
We construct this term as follows:
\begin{enumerate}
    \item \textit{Pairwise Distance Matrix:}  
    \[
    Q_{ij} = \|\tilde{C}_i - \tilde{C}_j\|_2,
    \]
    for $\tilde{C_i}$, $\tilde{C_j}$, the positions of agents $i$, $j$, respectively.
    \item \textit{Adjacency Matrix:}  
    \[
    A_{ij} = \sigma\left(-\frac{1}{r} \left(Q_{ij} - r\right)\right),
    \]
    for $r$ the interaction range, $\sigma(x) = {1}/({1 + e^{-x}})$, and $\forall i, A_{ii}=0$.
    \item \textit{Diagonal Degree Matrix:}  
    \[
    D_{ii} = \sum_{j} A_{ij}\ .
    \]
    \item \textit{Normalized Graph Laplacian:}  
    \[
    \tilde{L} = D^{-\frac{1}{2}} L D^{-\frac{1}{2}}\ .
    \]
    Where $L = D - A$.
    \item \textit{Heat Diffusion Filter:}  
    \[
    H = e^{-\tau \tilde{L}}\ ,
    \]
    where $\tau$ is the diffusion time parameter.
    \item \textit{Diffusion Evaluation:} 
    Obtain the smoothed signal for each node 
    , using $\delta_i$ which is 1 at node $i$ and 0 elsewhere):
    \[
    \forall i,\ {h}_i = H \delta_i\ , 
    \]  
    and evaluate the average coefficient of variation ($CV$) of $h$ across all $N$ nodes:  
    \[
    \mathcal{L}_{\text{diff}} = -\frac{1}{N} \sum_{i=1}^{N} \frac{\text{std}({h}_i)}{\text{mean}({h}_i) + \epsilon}\ ,
    \]
    where a small $\epsilon$ is used to avoid division by zero. Since we seek to maximize $CV$, we minimize its negative value. 
\end{enumerate}
\end{enumerate}
The effective loss function $L_{\textsc{SwaGen}}$ combines the four loss components, each weighted by an independent coefficient (all hyperparameters, see~\ref{app:loss_weighting}): 
\begin{align*}
    L_{\textrm{SwaGen}} =  \delta \mathcal{L}_{\textrm{KLD}} + \phi \mathcal{L}_{\textrm{rec}} + \beta \mathcal{L}_{\textrm{DOD}} + \gamma \mathcal{L}_{\textrm{diff}} \ .
\end{align*}
Of note, a fifth loss term of repulsion can be added to prevent agents from being mapped to the same point, modeling the avoidance range $\alpha$ (see~\ref{sec:swarm-model},~\ref{app:repulsion}).

\subsection{Optimizing the DOD-diffusion plane} \label{sec:gen-model-loss}

\begin{figure}[htb!]
  \begin{center}
    \includegraphics[width=0.75\linewidth]{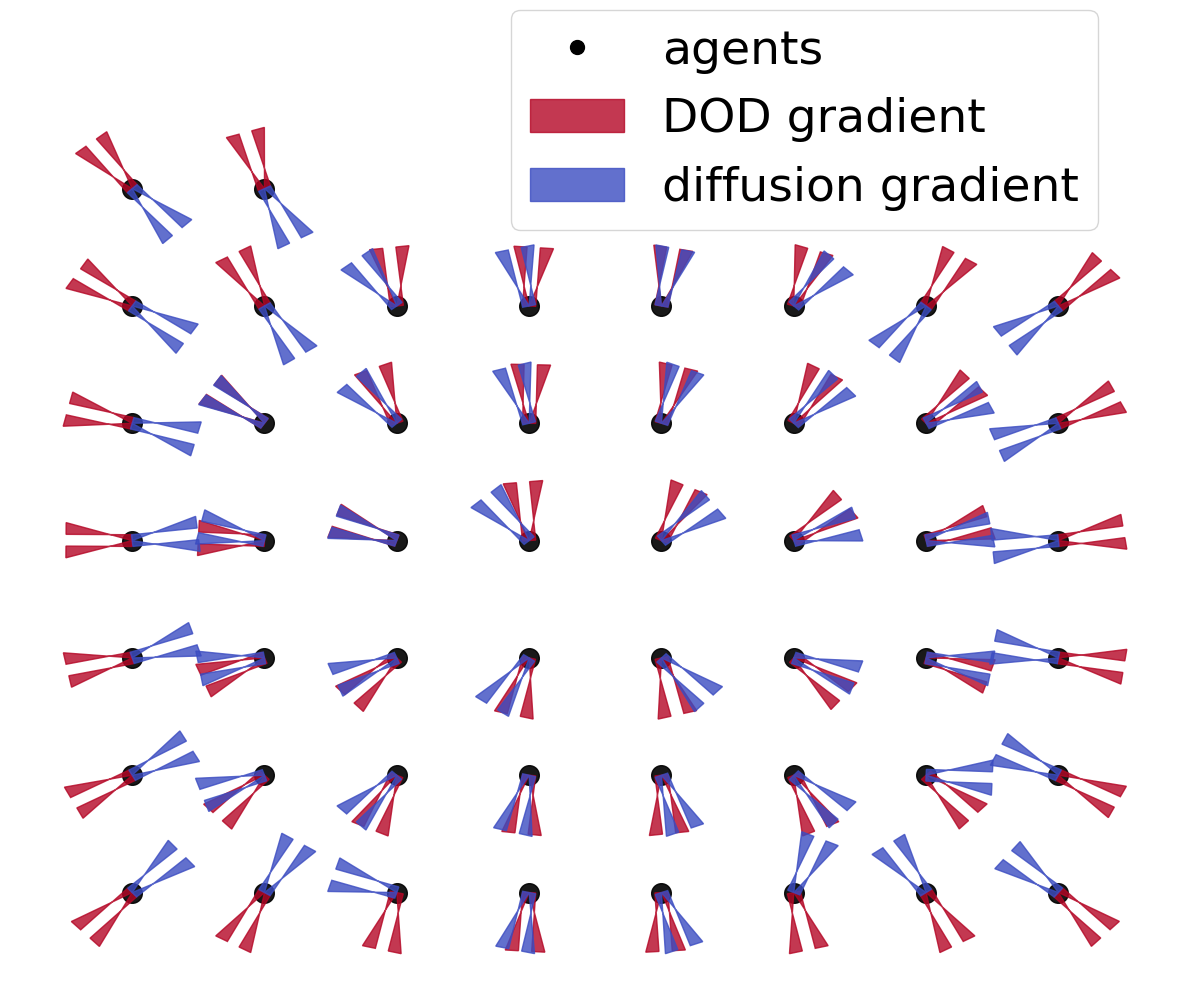}
  \end{center}
  \caption{\textbf{The models' gradients depict the DOD-diffusion trade-off}  Scatter plot of the swarm with the direction of the diffusion term gradient (blue) and the DOD term (red) for each agent. The gradients are negatively correlated without completely canceling, with Pearson correlation of $-0.82$.}
  \label{fig:gen_gradients}
\vspace{-2mm}
\end{figure}

As the DOD and diffusion are two contradictory properties of the swarm, we ought to validate that simultaneous optimization is plausible. Comparing the magnitude and direction of the terms' gradient, we could verify that they induce non-vanishing gradients, namely allowing for optimization (\fref{fig:gen_gradients}). Across all evaluated shapes, we find that gradients are negatively correlated but do not fully cancel each other out (Pearson correlation coefficient between the terms' gradients for each shape; v-formation: $-0.64$, arrow: $-0.78$, rectangle: $-0.82$). 

We next examine the two extreme points within the DOD-diffusion trade-off by shifting their relative weight (see Appendix~\fref{fig:extreme_loss}); In the case of high DOD loss, all agents are mapped to nearly the same point (the minimal DOD). In contrast, with high diffusion loss, agents are mapped far apart from one another, resulting in a graph where each agent is disconnected from the rest, preventing diffusion ability. 

\subsection{Training procedure}
\label{sec:gen-model-training}
The model was trained to predict optimized agent positions from the initial configurations using stochastic gradient descent with the Adam optimizer ($\text{learning rate}=0.01$) for $25$ epochs.
As a training set we generated $1000$ swarm configurations on a 2D grid (of size $1000 \times 1000$ pixels, see~\ref{app:graph-construction}). The number of agents, $N$, is sampled from a predefined range $\left[N_{\text{min}}, N_{\text{max}}\right]$, set to $N_{\text{min}}=40$ and $N_{\text{max}}=50$ in the presented results. The positions of the agents are sampled independently (see~\ref{app:data-gen},~\ref{app:more_training_details})

\section{Experiments} \label{sec:experiments}

\begin{figure}[htb!]
  \begin{center}
    \includegraphics[width=0.9\linewidth]{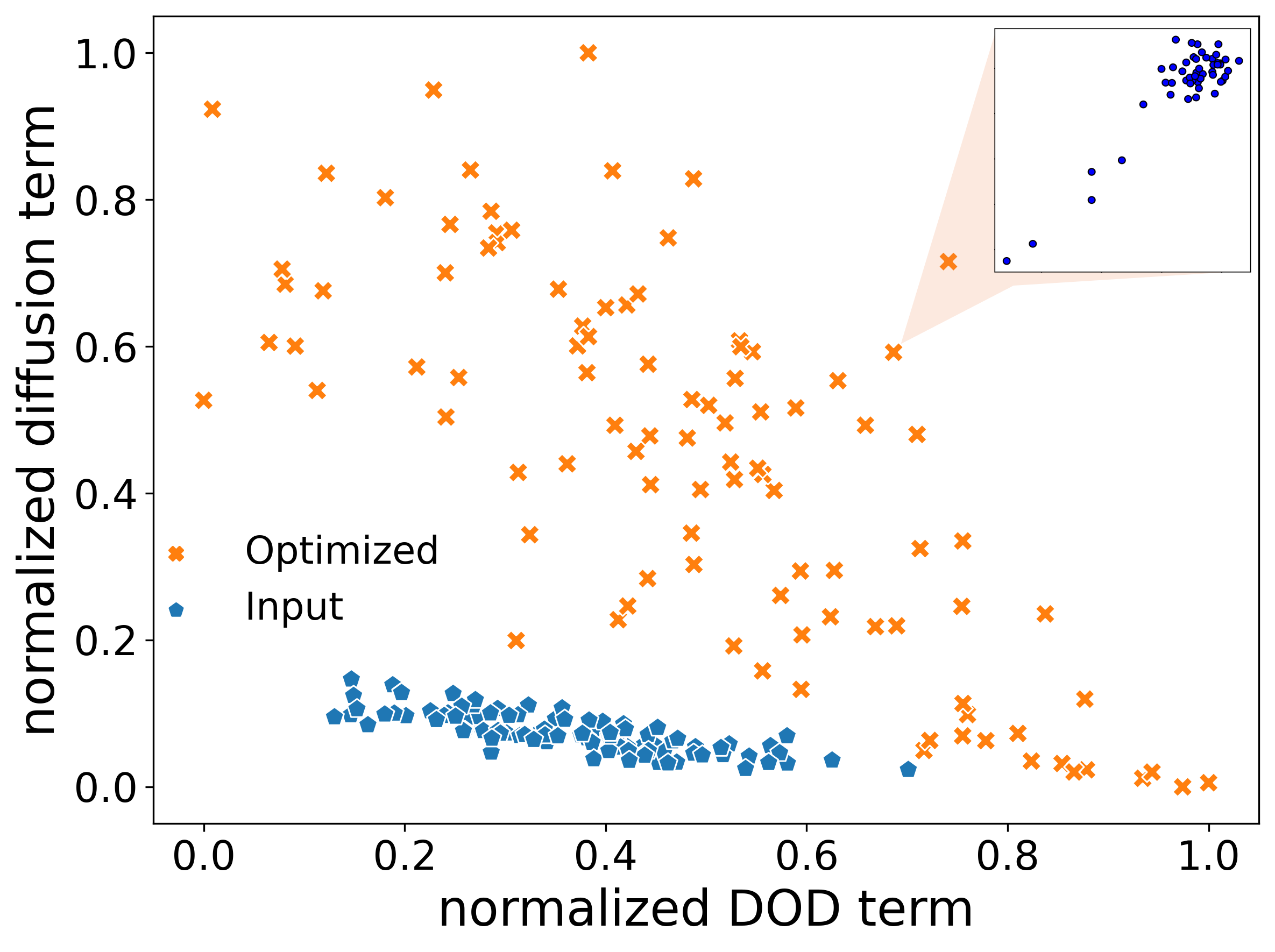}
  \end{center}
\caption{\textbf{Visualization of the trade-off considering random and optimized swarm spatial configurations.}  
The normalized diffusion ($y$-axis) compared to the normalized DOD ($x$-axis). For comparison, values are normalized to the range $[0,1]$ such that the top-right corner $(1, 1)$ is the optimal value (minimal DOD and diffusion).  Blue points represent the random swarm initialization, and orange crosses are the optimized swarm configurations.}
  \label{fig:optimization_analysis}
\vspace{-4mm}
\end{figure}

\subsection{\gnn~recovers optimal spatial configurations}\label{sec:optimization-analysis}

We first consider the DOD and diffusion terms, of random (model input) and optimized (model output) configurations (considering $100$ random swarm initializations, see~\ref{sec:gen-model-des}, ~\ref{app:data-gen}). 
Terms are normalized to the range of $[0,1]$ such that $1$ is the optimal value of each--minimal DOD and minimal diffusion.
Visualizing these results over a 2D plane (DOD-diffusion), we observe that optimized points (representing spatial configurations) improve the random initializations, according to both trade-off terms, closer to the optimal top right corner (\fref{fig:optimization_analysis}, ~\fref{fig:scatter_shapes_loss}, Table~\ref{tab:model_loss}).

\vspace{-4mm}
\begin{table}[h]
\caption{\textbf{Swarm’s loss terms improvement following optimization}. The mean normalized DOD and Diffusion terms for different spatial configurations reported for $250$ independent configurations.}
\vskip 0.15in
\begin{center}
\begin{small}
\begin{sc}
\begin{tabular}{lcccc}
\toprule
Shape/Mean term \% & DOD & Diffusion   \\
\midrule
Input shape & 0.364 & 0.074  \\
Optimized shape & 0.496 & 0.446  \\
\bottomrule
\end{tabular}
\end{sc}
\end{small}
\end{center}
\label{tab:model_loss}
 \vskip -0.1in
\end{table}

\subsection{The ``kite" motif: a novel durable configuration}
Analyzing the geometric properties of the optimized swarm configurations, we identified a common ``kite'' motif--a block of agents and a curved line of agents (\fref{fig:optimization_analysis}).

To systematically assess the kite motif and compare it to alternative formations 
we defined a theoretical model of a kite (\fref{fig:shapes visualization}d, see~\ref{app:kite}),  
and validated it is a good representation of the optimized configurations (\tref{tab:gw}). The validation was performed using the L2 Gromov-Wasserstein (GW) distance~\cite{memoli2011gromov} between optimized configurations and the kite along with the common spatial configurations (see~\ref{sec:analytical_model}, ~\ref{app:gw}).

\begin{table}[h]
\caption{\textbf{Validation that \gnn~optimized output is represented by the kite model.} The mean GW distance between optimized spatial configurations and a set of spatial configurations--kite, V-formation, Arrow, and Rectangle. Results are reported for $100$ independent random configurations optimized with \gnn.}
\label{tab:gw}
\vskip 0.15in
\begin{center}
\begin{small}
\begin{sc}
\begin{tabular}{lcccc}
\toprule
Configurations &    GW  \\
\midrule
\textbf{Kite} & \textbf{0.0259}   \\
V-formation & 0.0946  \\
Arrow & 0.0612   \\
Rectangle & 0.0911  \\
\bottomrule
\end{tabular}
\end{sc}
\end{small}
\end{center}
\vskip -0.1in
\end{table}

Evaluation of the durability of the kite to predation, using the same predation modeling approach previously described (see~\ref{sec:res-predation}), revealed that it improves the observed trade-off between detection and survival (\fref{fig:scatter_shapes_tradeoff}). The kite is closest to the optimal values $(100\%, 100\%)$--namely it is hardest to detect, and improves the overall survival rate of the swarm. Importantly, it also presents an improvement in the mean number of surviving agents (\tref{tab:shapes_survive}).

At last, evaluation of the analytical 2D DOD-diffusion plane (see ~\ref{sec:analytical_model}) revealed a similar trend--the kite was found above the diagonal of the plane, resembling the trade-off front ($(0,1)$-$(1,0)$), therefore presenting an improvement of it (\fref{fig:scatter_shapes_loss}). 

\begin{figure}[htb!]
  \begin{center}
    \includegraphics[width=0.9\linewidth]{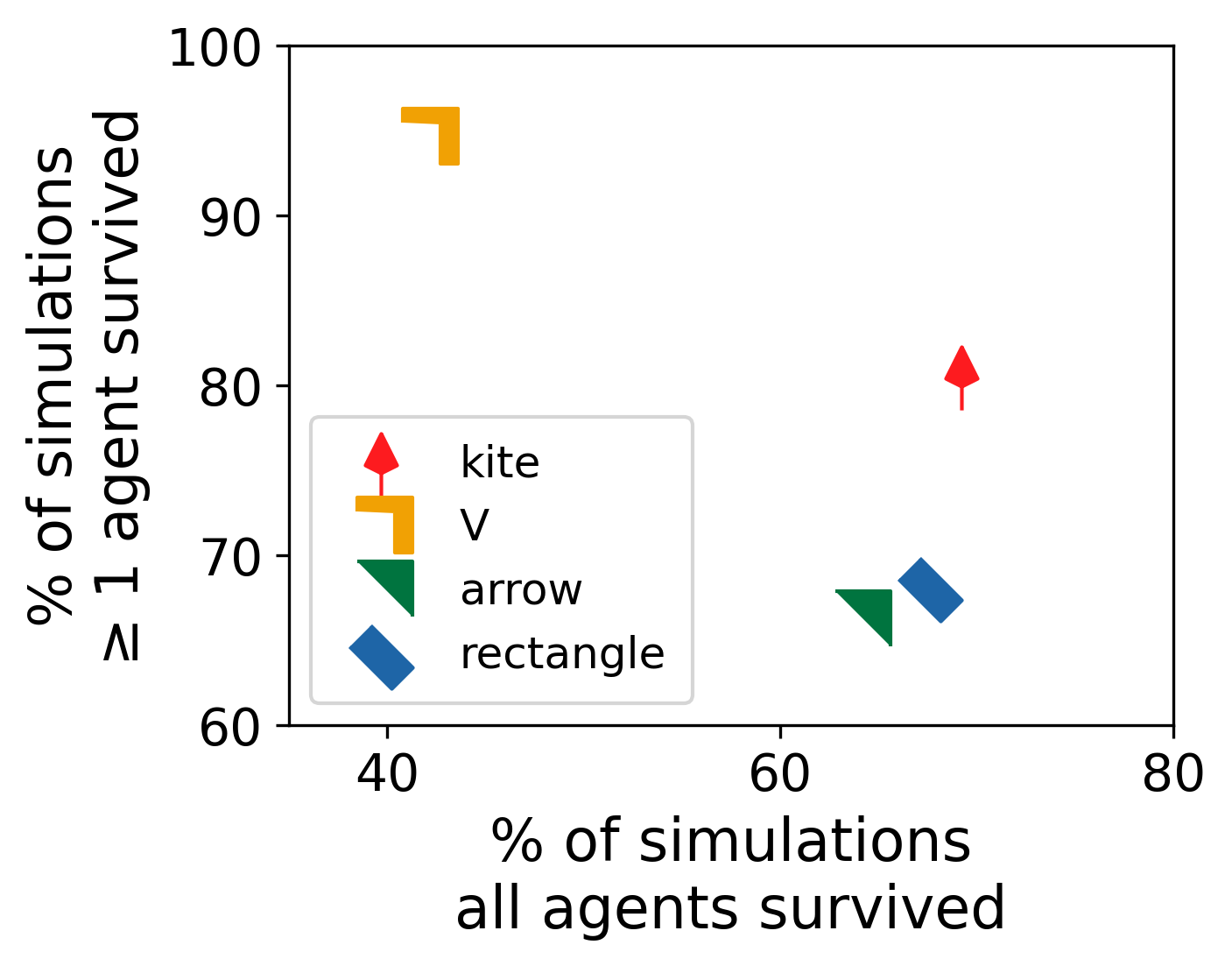}
  \end{center}
  \caption{\textbf{The kite improves the swarms' durability trade-off.}
Evaluation of the percentage of simulations in which the swarm avoided detection (none of the agents were predated; x-axis) compared to those that avoided extinction (at least one agent survived; y-axis) for the different spatial configurations (see~\ref{sec:predator-model}). Results are reported for $250$ independent simulations.}
\label{fig:scatter_shapes_tradeoff}
\vspace{-4mm}
\end{figure}

\begin{table}[h]
\caption{\textbf{Agents' survival rate is increased by the kite} The mean percentage of surviving across model simulations (Results provided for $250$ independent simulations)}
\label{tab:shapes_survive}
\vskip 0.15in
\begin{center}
\begin{small}
\begin{sc}
\begin{tabular}{lcccc}
\toprule
Shape &    Percentage of surviving agents \\
\midrule
\bf{Kite} & \bf{71.14\%}   \\
V-formation & 66.50\%  \\
Arrow & 64.42\%   \\
Rectangle & 67.61\%  \\
\bottomrule
\end{tabular}
\end{sc}
\end{small}
\end{center}
\vskip -2mm
\end{table}

\section{Discussion}\label{sec:discussion}

In this work, we used GSP theory to develop an analytical approach for studying the impact of a swarm's spatial configuration on its collective behavior. A key application of this approach was modeling and understanding swarm responses to external perturbations, such as predation. This analysis revealed a fundamental detectability-durability trade-off: configurations that render a swarm harder to detect may also make it more vulnerable once an attack commences, and vice-versa.

This inherent trade-off can be defined as an optimization task, which we addressed by introducing \gnn, a GNN-based generative framework. The \gnn~loss function was specifically designed to mediate this trade-off, incorporating terms that capture contradictory properties of the swarm's configuration. \gnn~generated novel swarm configurations, most notably the "kite" structure, which demonstrated superior performance with respect to this trade-off. These optimized configurations concurrently reduced the swarm's detectability and enhanced its survival rate post-detection.
 
 We believe that \gnn, as a tool for designing durable swarms, along with the elucidation of the detectability-durability trade-off and the supporting theoretical insights, could be used in future applications to guide the engineering of artificial swarms—such as drone formations or robotic teams—that are demonstrably more resilient to diverse classes of perturbations. Furthermore, our findings may shed light on the underlying mechanisms that drive the formation of  spatial configurations observed in natural swarms, from schools of fish to flocks of birds, as they navigate complex environments and predation risks.

Finally, we acknowledge the limitations of this study. Providing a complete quantification of the landscape of optimized spatial structures generated by \gnn~remains a complex endeavor, within this, our evaluations, alongside the consistent emergence of the kite motif, provide strong support for the utility of the optimized landscape. 
Additionally, our current work focuses on static initial spatial configurations. Future research could extend to dynamic graph environments, enabling real-time learning and adaptation of durable structures, mimicking the responsive nature of many natural and engineered swarms.


\section*{Code}
The code is publicly available at: 
\url{https://github.com/nitzanlab/SwaGen}

\section{Acknowledgements}
We thank Roy Friedman and Nitzan lab members for the help and feedback.


\bibliography{refrence}
\bibliographystyle{plain}

\clearpage
\appendix

\renewcommand{\thetable}{A\arabic{table}}
\renewcommand{\thefigure}{A\arabic{figure}}

\onecolumn
\section{Appendix}

\subsection{The analytical model}

\subsubsection{Unit disk graph for swarm modeling}\label{app:graph-construction}
We model the swarm as an (unweighted and undirected) unit disk graph, where agents (vertices) are connected by edges if disks around them intersect. We set the radius of these disks to be equal to the interaction (sensing) range $\rho$ as proposed by~\cite{mohamed2021graph}. 

Formally, let $G = (V, E)$ denote the graph where  $V$ is the set of agents and $E$ represents the set of edges between them. An edge $e_{ij} \in E $ exists between agents $i$ and $j$ if the Euclidean distance between them $d(i, j) \leq \rho$.

\subsubsection{Initial swarm spatial configurations for comparative analysis}\label{app:initial-shape}

The spatial configuration of the swarm plays a central role in dictating the swarm's dynamics, and as we will show, affects its durability to external perturbations. To this end, we consider spatial configurations that have been previously studied in this regard, the v-formation, a common structure in bird flocks~\cite{portugal2020bird}, an arrow configuration~\cite{vicsek2012collective, hoang2018angle}, and a rectangular configuration of agents~\cite{vicsek2012collective} (\fref{fig:abstract}a). 

For simplicity, the v-formation is set to a $90^\circ$ angle and the arrow formation is structured as a right triangle, both are oriented in the preferred direction of the informed agents. Similarly, the rectangle is oriented such that its longer side is perpendicular to the preferred direction. In addition, as a "null" spatial configuration, we included a random structure, positioning agents randomly within a square, matching the length of its side to the length of the v-structure edges. Across all initial configurations, all agents are located such that there is a minimal distance between closest neighbors, $d$.

\subsubsection{Simulation implementation details}\label{app:simulation-details}
Based on the swarm model (see~\ref{sec:swarm-model}) and the predator model (see~\ref{sec:predator-model}) we defined the following simulation to study the swarm dynamics: the initial swarm configuration (see~\ref{app:initial-shape}) is set at the bottom-left corner of a grid with periodic boundary conditions. The simulation then progresses for $T$ time point. At each time point we update the location and direction of movement of all agents. In addition, in case a predator is included in the simulation, the agents' status (living/non-living) is updated (\tref{tab:simulation_attributes}).

The simulation models a group of agents interacting within a 2D bounded environment (1000$\times$1000 points) with periodic boundary conditions.  In each iteration, the simulation synchronously updates the position and velocity vectors of each agent.
The simulation parameters are reported in~\tref{tab:simulation_attributes}.

\begin{table}[h!]
\centering
\begin{tabular}{|p{4cm}|p{7cm}|p{4cm}|}
\hline
\textbf{Attribute}          & \textbf{Description}                                                                        & \textbf{Value/Range}                                  \\ \hline
\texttt{Position}            & Position of the agent, $0 \leq x, y \leq 1,000$                                             & Variable                                              \\ \hline
\texttt{Direction}           & Direction vector, normalized such that $x + y = 1$                                          & Variable                                              \\ \hline
\texttt{Live}                & Status of agent (whether predated)                                                          & Variable (Boolean)                                    \\ \hline
\texttt{Informed}            & Whether the agent has a preferred direction                                                 & Constant (Boolean)                                    \\ \hline
\texttt{Predator}            & Whether the agent is a predator                                                             & Constant (Boolean)                                    \\ \hline
\texttt{Speed}               & Speed at which agents move                                                                  & 2                                                     \\ \hline
\texttt{Interaction\_range}  & Range within which agents interact with each other                                          & $\rho = 50$                                           \\ \hline
\texttt{Weighting\_term}     & Weighting for how much an informed agent follows its preferred direction                    & $\omega = 0.5$                                        \\ \hline
\texttt{Informed\_fraction}  & Fraction of informed agents                                                                 & 0.2                                                   \\ \hline
\texttt{Avoidance weighting} & Weighting of the avoidance term                                                             & 1                                                     \\ \hline
\texttt{Attraction weighting}& Weighting of the attraction term (preserving the initial shape)                             & 0.1                                                   \\ \hline
\texttt{Direction alignment weighting} & Weighting of the alignment term                                                  & 1                                                     \\ \hline
\texttt{Noise\_std}          & Standard deviation of Gaussian noise added to the preferred direction                       & Variable                                              \\ \hline
\texttt{Detection range}     & Maximum distance a predator can detect prey (equal to $\rho$)                               & 50                                                    \\ \hline
\texttt{Killing range}       & Maximum distance a predator can eliminate prey                                              & 5                                                     \\ \hline
\texttt{Predator speed}      & Speed of the predator                                                                       & 3                                                     \\ 
\hline
\texttt{Minimum distance}      & Initial minimum distance between agents                                                                       & $d=10$ ($\%20$ of $\rho$)                                                     \\ \hline
\end{tabular}
\caption{Simulation agent and global attributes}
\label{tab:simulation_attributes}
\end{table}

\begin{table}[t]
\caption{Mean coefficient of variation following diffusion by heat filter. Reported are the mean $CV$ values for different spatial configurations across a range of $\tau$ values.}
\label{table:mean_cv}
\vskip 0.15in
\begin{center}
\begin{small}
\begin{sc}
\begin{tabular}{lcccc}
\toprule
$\tau$  & V-formation & Arrow & Rectangle \\
\midrule
50  & 3.17 & 0.78 & 0.66  \\
100 & 2.63 & 0.43 & 0.30  \\
150 & 2.36 & 0.26 & 0.14  \\

\bottomrule
\end{tabular}
\end{sc}
\end{small}
\end{center}
\end{table}

\clearpage

\subsection{The \gnn~model}
\subsubsection{Repulsion loss term}\label{app:repulsion}
The repulsion loss term is designed to prevent agents to be mapped into the same point (modeling the avoidance range-$\alpha$ which is described in~\ref{sec:swarm-model}). It penalizes outputs of the model where any pair of agents is closer than the avoidance range.

\[
\mathcal{L}_{\text{repulsion}} = \frac{1}{N^2} \sum_{i=1}^{N} \sum_{j=1, j \neq i}^{N} \max(0, \alpha - Q_{ij})
\]

\begin{itemize}
    \item $N $: The number of agents.
    \item $Q_{ij} $: The Euclidean distance between agents $i $ and $j $.
    \item $\alpha $: The avoidance range.
\end{itemize}

\subsubsection{Model input}\label{app:model-input}
The model input consists of four features.

\textbf{Swarm positions:} The positions of the agents are represented by a matrix 
    \[
    C \in \mathbb{R}^{N \times d},
    \]
    where $N$ is the number of agents, and $d$ is the dimension of the space (here we used only $d = 2$ for a 2D plane). 

\textbf{Number of neighbors:} For each agent $i$, the number of neighbors is given by
    \[
    k_i = \sum_{j \in V \setminus \{i\}} \#(d_{ij} \leq r),
    \]
    where $d_{ij}$ is the Euclidean distance between agents $i$ and $j$, $\rho$ is the interaction range, and $\mathbb{\#}(\cdot)$ is indicator function.~\cite{mateo2017effect} showed that the number of interacting neighbors could influence the swarm's response to perturbations.

\textbf{Distance from the center of mass:} Let the center of mass of the swarm be 
    \[
    \mathbf{c} = \frac{1}{N} \sum_{i=1}^N \mathbf{c}_i,
    \]
    where $\mathbf{c}_i \in \mathbb{R}^d$ is the position of agent $i$. The distance of agent $i$ from the center of mass is
    \[
    d_i = \|\mathbf{c}_i - \mathbf{c}\|_2.
    \]

\textbf{Angle from the center of mass:} The angle of agent $i$ relative to the center of mass is
    \[
    \theta_i = \arctan\left(\frac{\mathbf{y}_i - \mathbf{y}}{\mathbf{x}_i - \mathbf{x}}\right),
    \]
where $\mathbf{x}_i$ and $\mathbf{y}_i$ are the $x$ and $y$ coordinates of agent $i$, and $\mathbf{x}$, $\mathbf{y}$ are the corresponding coordinates of the center of mass.
We used the radius and the angle from the center of mass to allow the neural network to learn and generate motifs that are not dependent on a specific location on the grid.

\subsubsection{Input generation}\label{app:data-gen}
We used a simple sampling procedure to create random swarms localized in different locations along the 2D grid. It takes the following parameters:
\begin{itemize}
    \item $N$: Number of agents.
    \item  $\mu$,  $\sigma^2$: Mean and variance for the Gaussian distribution of the center of positions.
    \item  $\mu_r$,  $\sigma_r^2$: Mean and variance for the Gaussian distribution of position range.
\end{itemize}

The process is as follows:

1. The mean position values $\mu_x $ and $\mu_y $ are sampled from normal distributions centered at $\mu $ with standard deviation $\sqrt{\sigma^2} $.
2. The range values $\sigma_x $ and $\sigma_y $ are sampled from normal distributions centered at $\mu_r $ with standard deviation $\sqrt{\sigma_r^2} $.
3. The positions of each agent are sampled from a uniform distribution within the ranges defined by the mean and range values:
\[
x_i \sim \mathcal{U}\left(\mu_x - \frac{\sigma_x}{2}, \mu_x + \frac{\sigma_x}{2}\right),
\quad y_i \sim \mathcal{U}\left(\mu_y - \frac{\sigma_y}{2}, \mu_y + \frac{\sigma_y}{2}\right).
\]
4. The positions are clipped to ensure they stay within the bounds of the 2D space ($0 \leq x_i, y_i \leq 1000$).

The output is a 2D array with the generated positions of all agents.

\subsubsection{Training details}\label{app:more_training_details}
\textbf{Input normalization}:
In the first layer, the input $\mathbf{x} $ is normalized and scaled as follows:
\[
\mathbf{x} = \frac{\mathbf{x}}{1000} - 0.5
\]
This scales the input by $1000 $ and centers it around 0.

In the last layer, the output is scaled back:
\[
\text{output} = \sigma(\text{linear}(\mathbf{x})) \times 1000
\]
where $\sigma $ is sigmoid activation.

\textbf{Package details}:
The neural network is based on PyTorch-geometric package (version 2.7.0, \citealp{fey2019fast}).

\textbf{Decay factor}:
 A decay factor was included for the reconstruction term  ($\lambda=0.9$):
\[
\phi_{\text{epoch}} = \phi \cdot \lambda^{\text{epoch}}
\]
where $\phi$ is the initial reconstruction coefficient,$\text{epoch}$ is the current training epoch. 
This decay factor ensures that global constraints (e.g. converging to the central region of the grid) are established before fine-tuning local structural details based on the diffusion and the domain of danger terms.

\subsubsection{Loss weighting}\label{app:loss_weighting}

Each loss term is designed to be weighted per agent or grid point. For the DOD term, we take the mean over all grid points. For the diffusion term, we use the coefficient of variation (CV), normalizing variance by the mean to approximate per-agent variability. The reconstruction loss is simply the MSE between initial and optimized positions.
We add hyperparameters to scale each loss term in the integration of all terms.  
The reconstruction term is weighted by $\phi_{\text{epoch}}$ as described above, and the KLD term hyperparameter, $\delta$, is set to $1$.
Balancing $\beta$ (DOD term) and $\gamma$ (diffusion term) can be done by assessing the loss plane (\fref{fig:optimization_analysis}).

\subsubsection{Kite-shaped swarm}\label{app:kite}

We generate a kite-like spatial configuration of $N$ agents which includes a rectangular block of points at the top and a curved line below it:

\begin{itemize}
    \item \textbf{Top Block:} A rectangular grid of $\lfloor N*3 / 4 \rfloor$ points, spaced by $\alpha$.
    \item \textbf{Curved Line:} The remaining points form a curved line below the block. To create a non-linear curve we used the following formula for the $i$ agent location: 
    \[
    x_i = -i \cdot \alpha \cdot \cos(f \cdot i), \quad y_i = -i \cdot \alpha,
    \]
    where $f$ controls the curvature.
    
\end{itemize}
This structure combines order (in the block) with natural curvature (in the line) to resemble a kite.

\subsubsection{Gromov-Wasserstein distance calculation}\label{app:gw}

To evaluate the similarity between the output of our generative model and the spatial agent configurations models, we computed the Gromov-Wasserstein (GW) distance, which compares two metric spaces based on their pairwise distance matrices. Specifically, we used L2 metric as the distance metric. 
Given two sets of agent positions, $\mathbf{X}$ and $\mathbf{Y}$, we define their Euclidean distance matrices:

\begin{equation}
C_1(i, j) = d(x_i, x_j), \quad C_2(i, j) = d(y_i, y_j),
\end{equation}

which are then normalized by their maximum values. We assume a uniform probability distribution over agents:

\begin{equation}
p = \frac{1}{N} \mathbf{1}, \quad q = \frac{1}{N} \mathbf{1}.
\end{equation}

The GW distance with a squared loss function is then computed as:

\begin{equation}
GW(C_1, C_2, p, q) = \min_{\gamma} \sum_{i,j,k,l} (C_1(i,j) - C_2(k,l))^2 \gamma(i,k) \gamma(j,l).
\end{equation}

We compare the optimized agent positions $\mathbf{X}_{opt}$ to predefined formations (kite, arrow, rectangle, and V-shape), computing GW distances to quantify their structural similarity:

\begin{equation}
d_{shape} = GW(\mathbf{X}_{shape}, \mathbf{X}_{opt}).
\end{equation}

\clearpage
\quad

\section{Proofs}\label{sec:proofs}
\subsection{Preliminaries.} Prior to the proofs, we provide definitions and facts regarding our system, relevant for the following proofs. 

We are considering a swarm of $N$ agents and predators, with the matrix $C(t) \in \mathbb{R}^{N \times d}$ and the vector $\mathbf{c}_{p}(t)$ representing their location respectively over the grid $A$. 
For simplicity, we drop the temporal dependence henceforth and relate to it implicitly.

\textbf{The Independence of Predators.} Predators act independently according to the predation rules.

\textbf{The Uniform Distribution Property.}  
 A predator's location $\mathbf{c}_{p}$ is uniformly distributed over the grid $A$. For any measurable subset $E \subseteq A $, the probability that the predator is located within $E$ is  given by
\begin{align*}
\mathbb{P}(\mathbf{c}_{p} \in E) = \frac{\mu(E)}{\mu(A)} \ ,
\end{align*}
for $\mu(\cdot)$ the area measure.

\textbf{The Domain of Danger (DOD).} The domain of danger defined as
\begin{align*}
    D = \bigcup_{i=1}^{N} B(\mathbf{c}_i, \rho)\ ,
\end{align*}
where $B(\mathbf{c}_i, \rho) $ is a disk of radius $\rho $ centered at the position $\mathbf{c}_i $ of agent $i$, and $N $ is the total number of agents.

\subsection{Detectability: predation probability is proportional to the DOD's magnitude} \label{sec:detect}

\begin{theorem} \label{th:probability} 

Given a swarm with $N$ agents advancing at speed $v_{\text{prey}}$ over a grid $A$, for a predator with interaction range $\rho$ advancing at a speed $v_p$ obeying  
\begin{align*}
v_p  > v_{\text{prey}} \ ,
\end{align*}
the probability of a detection within a time interval $\tau$, is  proportional to the relative area of the DOD, $D$, within the grid $A$,   
\begin{align*}
\mathbb{P}_{\text{pred}}(t+\tau) \approx \frac{\mu(D)}{\mu(A)} \ .
\end{align*}
And $\tau $  is bounded by:
\begin{align*}
 \tau \leq \frac{\rho}{v_p - v_{\text{prey}}}.
\end{align*}
\end{theorem}

\begin{proof}
The uniform distribution property implies that the probability that the predator is within the swarm's DOD,  $D \subseteq A$ is given by,
\begin{align*}
    \mathbb{P}(\mathbf{c}_{p} \in D) = \frac{\mu(D)}{\mu(A)} \ .
\end{align*}
Under the assumption that the predator is within the DOD of the swarm, we would like to obtain a bound on the time $\tau$ for a detection to occur.
The assumption implies that there exists an agent $m$ such that the distance $d_{m}$ between the predator and the agent obeys $d_{m} \leq \rho$.
 This distance bounds $\tau$, the minimal time period for a detection to occur:  
 \begin{align*}
      \tau = \frac{d_{m}}{v_p - v_{\text{prey}}} \leq  \frac{\rho}{v_p - v_{\text{prey}}},
 \end{align*}
where the bound is given by the extreme case $d_{m} \leq \rho$. 
\end{proof}

The presented \textit{Theorem} \ref{th:probability} ties between the detectability and the magnitude of the DOD of a given swarm.  
Given a large DOD, $D$, it is more likely that the predator will detect the swarm within a time interval $\tau$, since $D$ covers a large area over the grid and we assume a uniform distribution of the predator's location. 

\subsection{Durability: extinction of the swarm is correlated with the ease of diffusion in its graph} \label{sec:durability}
Following a detection event, the minimal distance the predator must travel to predate the swarm is given by:
\begin{align*}
    D_{\min}=\sum_{i=1}^{n-1} d\left(i, i+1\right) \ ,
\end{align*}
where $d\left(i, i+1\right)$ is the distance between two consecutive agents, and we assume optimal ordering--such that for agent $i$, $i+1$ is the closest neighbor. The corresponding minimal time to predate the entire swarm is then $T_{\min}  = D_{\min} / \left(v_p - v_{prey}\right)$. 

As the graph is defined on a euclidean manifold the predator's signal propagation kernel, $K$, can be approximated by a Guassian kernel (see ~\ref{sec:diff_graph}). This approximation implies that for large $D_{\min}$, a long time $\tau_{\min}$ is required for the predator to predate the entire swarm, scaling as $\tau_{\min} \sim D_{\min}$.
This shows that the diffusion ability is proportional to the minimal extinction time, both correlating with the minimal distance on the graph. 
To quantify this we conducted the following experiments--
A swarm of agents was generated randomly, on a 2D plane, in a consecutive manner; each new agent placed exactly at distance $\alpha$ from an existing agent. For each generated swarm, the diffusion loss term (see ~\ref{sec:gen-model-des}) and the mean minimal time for extinction ($T_{\min}$) were evaluated (considering each of the $N$ agents as the initial detection). We conducted this experiment with $1,000$ swarms with different randomly sampled $\alpha$ values, $\alpha \sim \mathcal{U}(5, 50)$ (\fref{fig:app_time_diff_correlation}). As expected, these values were greatly correlated (Spearman correlation $0.985$).

\begin{figure}[htb!]
\begin{center}
        \includegraphics[width=0.5\linewidth]{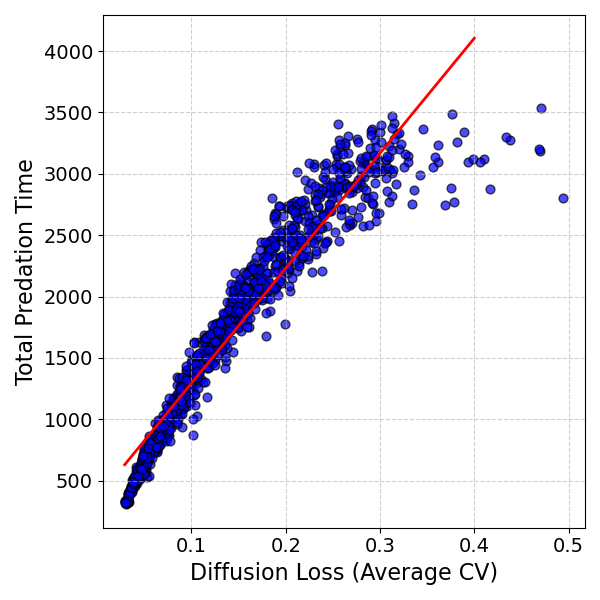}
\end{center}
  \caption{\textbf{The diffusion loss term correlates with the mean predation time}. The diffusion loss term (x-axis) and $T_{\min}$, the mean minimal time for extinction (y-axis) evaluated over $1,000$ swarm configurations. Each swarm was generated randomly, in a consecutive manner; each new agent placed exactly at distance $\alpha$ from an existing agent. The red line shows the linear fit, indicating that higher diffusion correlates with longer predation time (Spearman correlation: $0.985$).}
  \label{fig:app_time_diff_correlation}
\end{figure}

\subsection{The detectability-durability trade-off}\label{sec:tradeoff}
Sections \ref{sec:detect}, \ref{sec:durability} provide a theoretical connection between the swarm properties, detectability and durability, and the spatial configuration properties, the DOD and the ability to diffuse on the graph, and justifies the design of the \gnn~loss function (using the DOD and diffusion terms, to optimize the detectability-durability trade-off).  
Importantly, both quantities, the diffusion ability and the DOD are functions of the distances between agents within the swarm. More distant agents will induce a larger DOD magnitude, minimizing the overlap between independent agents' disks (with radius of $\rho$) and vice-versa. For the diffusion, we have derived that diffusion time is proportional to distances on the swarm.

\clearpage

\section{Supplementary figures}\label{sec:sup_figs}

\begin{figure}[htb!]
\begin{center}    
    \includegraphics[width=0.8\linewidth]{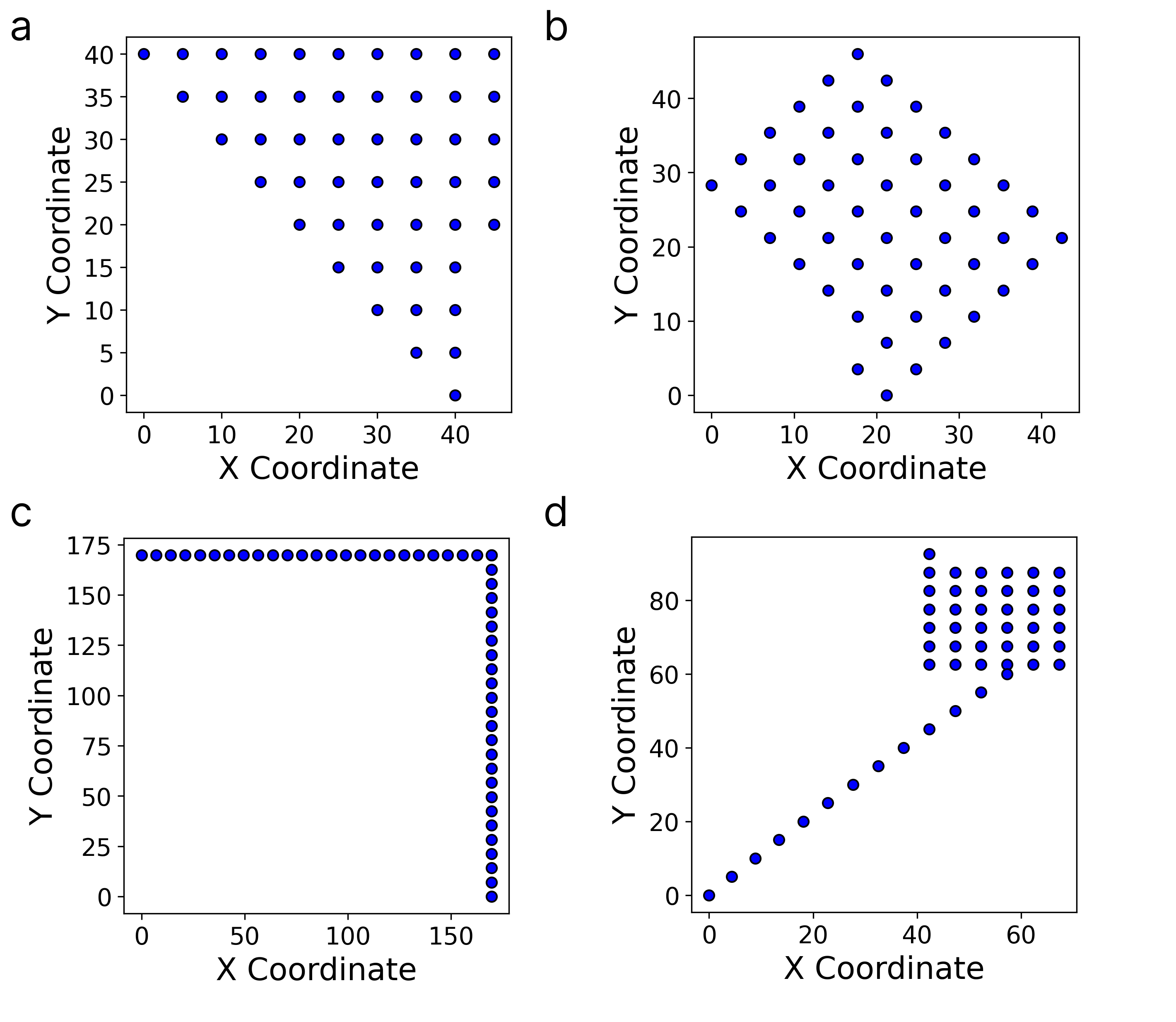}
\end{center}
  \caption{\textbf{Visualization of swarming configurations} (a) Arrow. (b) Rectangle. (c) V-formation. (d) Kite.}
  \label{fig:shapes visualization}
\end{figure}

\begin{figure}[htb!]
\begin{center}    
    \includegraphics[width=1\linewidth]{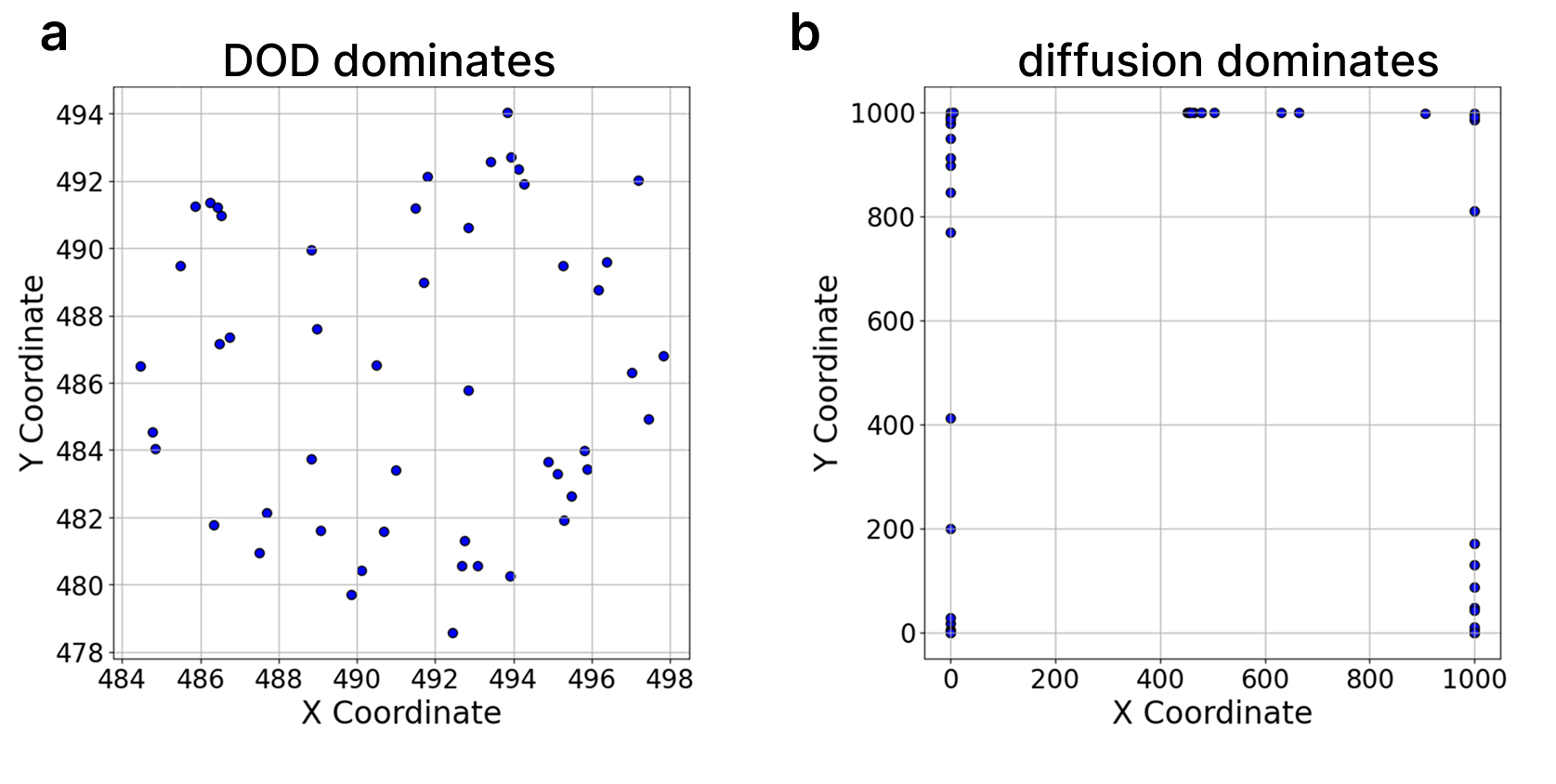}
    \end{center}
  \caption{\textbf{Value of the different loss terms in the SwaGen objective. }
  Evaluation of the Swagen model under dominance of one of the loss terms (a) domain of danger (DOD), and (b) diffusion. In (a) when the DOD dominates all agents converge to nearly the same point, minimizing the domain of danger but maximizing diffusion ability. Alternatively, in (b) When the diffusion dominates,, agents spread across the edges of the grid, minimizing diffusion ability but maximizing the domain of danger.
}
\label{fig:extreme_loss}
\end{figure}

\begin{figure}[htb!]
\begin{center}
    \includegraphics[width=1\linewidth]{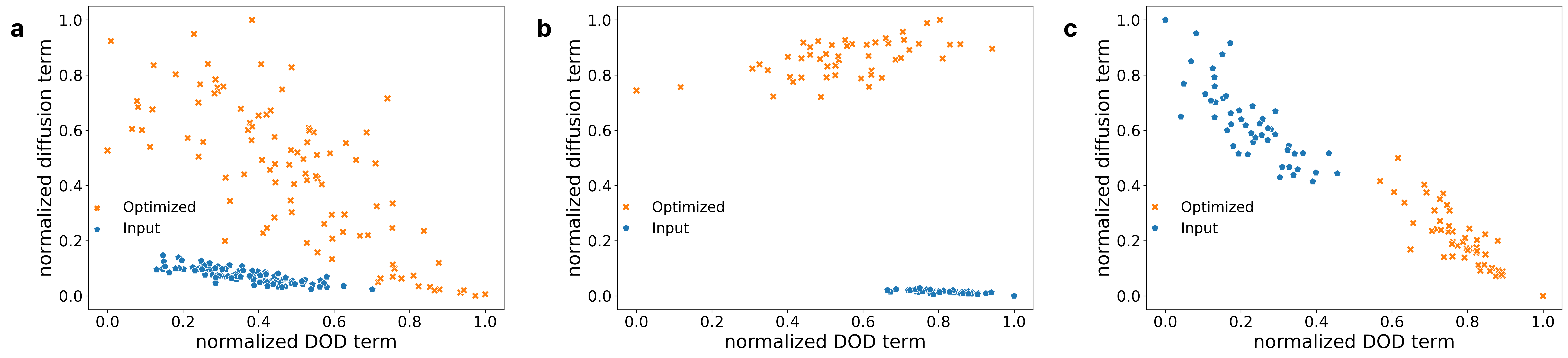}
    \end{center}
  \caption{\textbf{Trade-off visualization between random and optimized swarm configurations. }
  The normalized diffusion (y-axis) vs.normalized DOD (x-axis), both scaled to [0,1], with (1,1) as the optimal point (minimal DOD and diffusion). Blue dots represent random initialization, orange crosses indicate optimized configurations, (a) baseline, (b) diffusion loss weighted 2× the baseline, (c) DOD loss weighted 2× the baseline.
}
\label{fig:com_loss_plane}
\end{figure}

\begin{figure}[htb!]
  \begin{center}
    \includegraphics[width=0.5\linewidth]{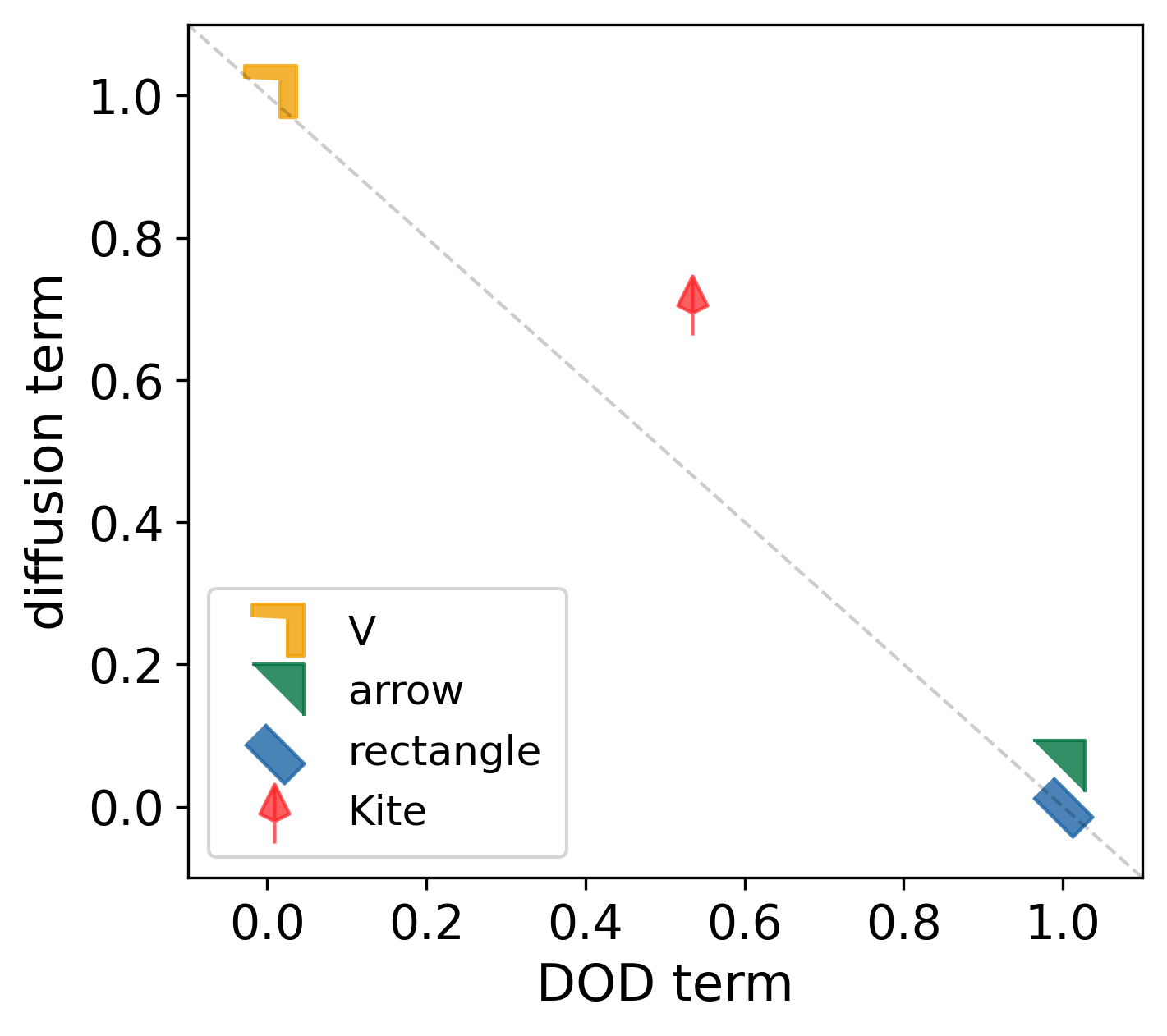}
  \end{center}
  \caption{\textbf{The "DOD-diffusion plane".}
The tension between the DOD and diffusion ability in different spatial configurations. Evaluation of the DOD (x-axis) and diffusion term (y-axis) normalized to the range of $[0,1]$. Evaluation for each spatial configuration is performed using $N=1000$ agents and $d=5$ (the distance between neighboring agents, see~\ref{app:simulation-details}).
To evaluate the diffusion we computed the mean values of the $CV$ for every possible source agent each configuration and considered heat kernels with $\tau=[50,100,150]$.}
\label{fig:scatter_shapes_loss}
\end{figure}

\begin{figure}[htb!]
  \begin{center}
    \includegraphics[width=1\linewidth]{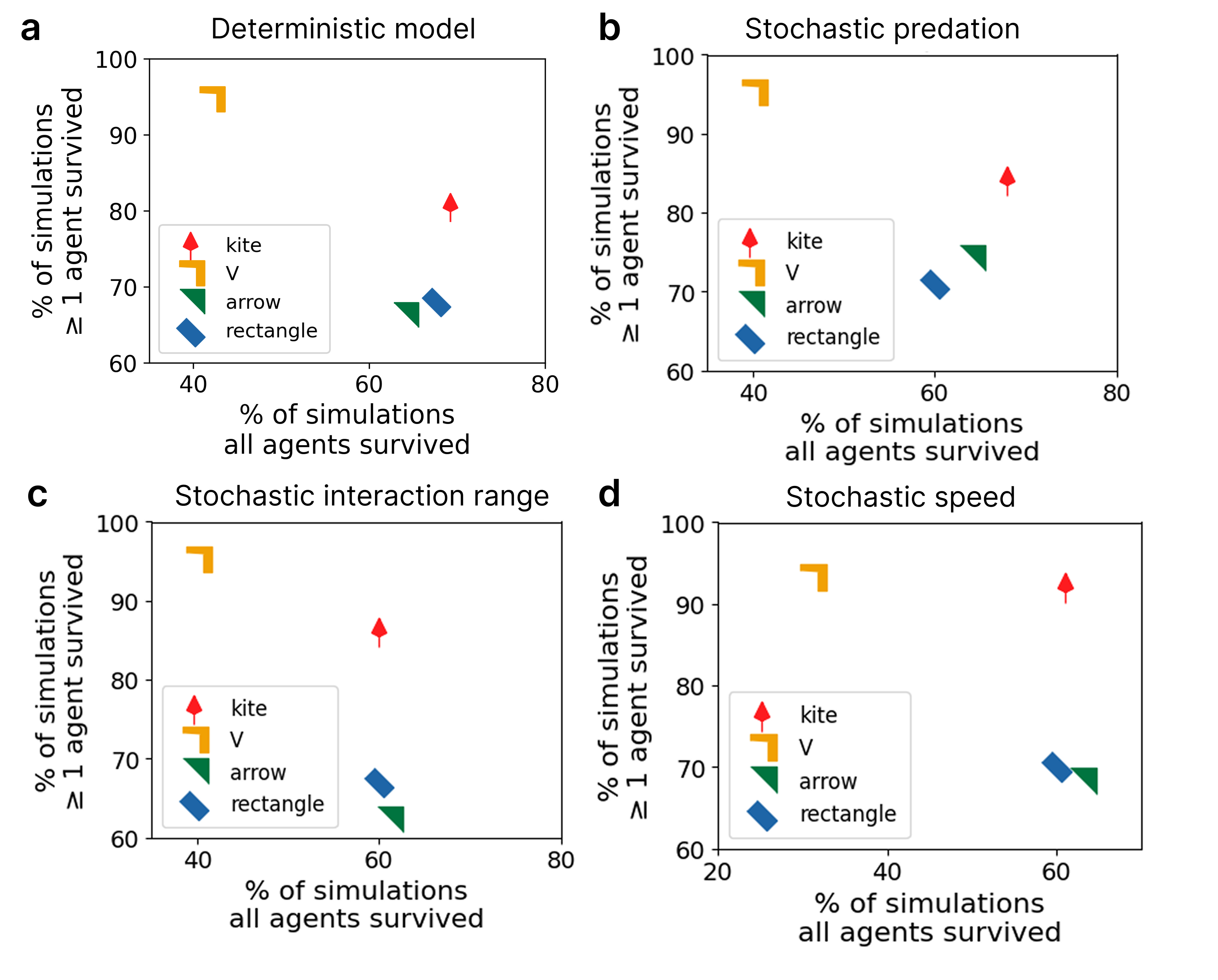}
    \end{center}
  \caption{\textbf{The swarm’s durability is maintained under stochastic conditions.}
  Evaluation of the percentage of simulations in which the swarm avoided detection (none of the agents were predated; x-axis) compared to those in which the swarm avoided extinction (at least one agent survived; y-axis) across different spatial configurations for: (a) a deterministic model, (b) under stochastic predation (when a prey is within the killing range of the predator, a predation occurs with  probability of 0.5), (c) stochastic interaction range (at each iteration the interaction range is sampled from a Gaussian distribution $\rho \sim \mathcal{N}(50, \sqrt{10})$), (d) 
  or with a stochastic speed (at each iteration the speed is sampled from a Gaussian distribution $s \sim \mathcal{N}(2, 1)$). Results are reported for 100 independent simulations per setting, the rest of the simulation parameters are identical to the parameters used for the results presented in the main text (see~\ref{sec:analytical_model}).
}
\label{fig:app_more_models}
\end{figure}


\end{document}